\newif\ifnotes
\notesfalse

\documentclass[11pt]{article}
\usepackage{fullpage} %[paper=letterpaper,margin=1in]{geometry}

\usepackage{xspace,amsmath,amsthm,amsfonts,amssymb,xcolor,afterpage,framed}
\usepackage{enumerate}
\usepackage{units} %for nicefrac
\usepackage{algpseudocode}
\usepackage{algorithm}
\usepackage{cancel}
\newcommand{\ignore}[1]{}

\definecolor{darkred}{rgb}{0.5, 0, 0}
\definecolor{darkgreen}{rgb}{0, 0.5, 0}
\definecolor{darkblue}{rgb}{0,0,0.5}

\newlength{\saveparindent}
\newlength{\saveparskip}
\newcounter{ctr}

\newcounter{ectr}

\renewcommand{\to}{\ensuremath{\rightarrow}}

\newcommand{\eps}{\ensuremath{\epsilon}}

\newcommand{\ceil}[1]{\ensuremath{\lceil{#1}\rceil}}
\newcommand{\floor}[1]{\ensuremath{\lfloor{#1}\rfloor}}

\newcommand{\abs}[1]{\ensuremath{\lvert{#1}\rvert}}

\newcommand{\set}[1]{\ensuremath{\{#1\}}}

\newcommand{\adv}{\ensuremath{\mathcal{A}}\xspace}

\def\zo{\{0,1\}}

\newcommand{\poly}{\ensuremath{\mathrm{poly}}\xspace}

\def\to{\rightarrow}

\newtheorem{theorem}{Theorem}[section]
\newtheorem{definition}[theorem]{Definition}
\newtheorem{claim}[theorem]{Claim}
\newtheorem{remark}[theorem]{Remark}

\newtheorem{lemma}[theorem]{Lemma}

\newtheorem{assumption}[theorem]{Assumption}

\newcommand{\bigtodo}[1]{%
    \begin{center}
    \fcolorbox{red!50!black}{red!20}{%
        \parbox{0.8\linewidth}{#1}%
    }
    \end{center}
}

\newlength{\protowidth}

\newcommand{\namedref}[2]{\hyperref[#2]{#1~\ref*{#2}}}
\newcommand{\subfigureref}[2]{\hyperref[#1]{Figure~\ref*{#1}#2}}

\definecolor{darkred}{rgb}{0.5, 0, 0}
\definecolor{darkgreen}{rgb}{0, 0.5, 0}
\definecolor{darkblue}{rgb}{0,0,0.5}

\def\bprot{\pi_{\text{beacon}}}
\def\mc{\mathcal}
\def\chain{\mathcal{C}}
\newcommand{\euler}{e}
\def\nofork{\mathcal{M}_\text{FL}}
\newcommand{\expect}{\mathrm{E}}
\def\dtv{\Delta}
\newcommand{\pr}{\mathrm{Pr}}
\newtheorem{overview}[theorem]{Overview}
\newtheorem{example}[theorem]{Example}
\newtheorem{observation}[theorem]{Observation}
\newtheorem{interpretation}[theorem]{Interpretation}

\newcommand{\fbeacon}{\ensuremath{\mathcal{F}_{\mathrm{beacon}}}\xspace}
\newcommand{\fchain}{\ensuremath{\mathcal{F}_{\mathrm{chain}}}\xspace}
\newcommand{\pibeacon}{\ensuremath{\pi_{\mathrm beacon}}\xspace}
\newcommand{\bc}{Bitcoin\xspace}
\newcommand{\csp}{\ensuremath{\kappa}}

\newcommand{\extr}{\textsc{ext}}
\newcommand{\hash}{{hash}}

\definecolor{webred}{rgb}{0.5,0,0}
\definecolor{webblue}{rgb}{0,0,0.8}
\PassOptionsToPackage{hyphens}{url}
\usepackage[colorlinks,citecolor=webblue,linkcolor=webred,backref,pagebackref]{hyperref}
\usepackage[hyphenbreaks]{breakurl}

\newcommand{\hs}[1]{\ifnotes{\color{red}{\bf [HS: #1]}}\fi}
\newcommand{\ariel}[1]{\ifnotes{\color{red}{\bf [Ariel: #1]}}\fi}
\newcommand{\iddo}[1]{\ifnotes{\color{red}{\bf [Iddo: #1]}}\fi}
\newcommand{\vassilis}[1]{\ifnotes{\color{red}{\bf [Vassilis: #1]}}\fi}
\newcommand{\aggelos}[1]{\ifnotes{\color{red}{\bf [{\bf Aggelos:} #1]}}\fi}

\ignore{

\newcommand{\hs}[1]{{\color{blue}{\bf [HS: #1]}}}

\newcommand{\vnote}[1]{{\color{darkred}{\bf [VZ: #1]}}}

}

\ignore{
\renewcommand{\iddo}[1]{}
\renewcommand{\hs}[1]{}
\renewcommand{\vassilis}[1]{}
\renewcommand{\aggelos}[1]{}
}

\ignore{%
\setlength{\textwidth}{16.4cm} \setlength{\textheight}{23.6cm}
\setlength{\oddsidemargin}{0.3cm} \setlength{\evensidemargin}{0.3cm}
\setlength{\topmargin}{-.5cm} \setlength{\headheight}{-.1cm}
\setlength{\headsep}{0cm} \setlength{\footskip}{.7cm}

}%%%

\begin{document}

\title{{\sc 
Bitcoin Beacon
}
}

\author{
Iddo Bentov\thanks{Research partly supported by ERC project CODAMODA.}\\
\footnotesize{Technion}\\
\footnotesize{\texttt{idddo@cs.technion.ac.il}}
\and
Ariel Gabizon\\
\footnotesize{Technion}\\
\footnotesize{\texttt{ariel.gabizon@gmail.com}}
%\and
%Aggelos Kiayias\thanks{
%University of Athens, 
%{\tt aggelos@kiayias.com}}
%\and
%Hong-Sheng Zhou\thanks{
%VCU\\
%Virginia Commonwealth University, 
%{\tt hszhou@vcu.edu}}
%\and
%Vassilis Zikas\thanks{
%ETH, 
%{\tt vzikas@inf.ethz.ch}}
\and
David Zuckerman\\
\footnotesize{University of Texas at Austin}\\
\footnotesize{\texttt{diz@cs.utexas.edu}}
}

\ignore{%%%%
\author{
Iddo Bentov\\
Technion\\
{\tt idddo@cs.technion.ac.il}
\and
Aggelos Kiayias \\
University of Athens\\
{\tt aggelos@kiayias.com}
\and
Ariel Gabizon\\
Technion\\
{\tt ariel.gabizon@gmail.com}
\and
Hong-Sheng Zhou\\
%VCU\\
Virginia Commonwealth University\\
{\tt hszhou@vcu.edu}
\and
Vassilis Zikas\\
ETH\\
{\tt vzikas@inf.ethz.ch}
\and
David Zuckerman\\
University of Texas at Austin\\
{\tt diz@cs.utexas.edu}
}
}%%%%%%%%
%\author{Anonymous submission}
\date{}

\maketitle

\thispagestyle{empty}

\begin{abstract}
We examine a protocol $\bprot$ that outputs unpredictable and publicly verifiable randomness, meaning that the output is unknown at the time that $\bprot$ starts, yet everyone can verify that the output is close to uniform after $\bprot$ terminates. We show that $\bprot$ can be instantiated via \bc under sensible assumptions; in particular we consider an adversary with an arbitrarily large initial budget who may not operate at a loss indefinitely. In case the adversary has an infinite budget, we provide an impossibility result that stems from the similarity between the \bc model and Santha-Vazirani sources. We also give a hybrid protocol that combines trusted parties and a \bc-based beacon.
\end{abstract}

\section{Introduction}

Consider multiple parties who wish to execute a high stake protocol that involves {\em public} randomness. For example, the parties may wish to elect one of them as a leader, in a way that allows anyone (including non-participants in the protocol) to verify that the elected leader did not corrupt the election process by offering bribes to other parties. If the source of the randomness can be tampered with, corrupt parties may try to influence this source in their favor. Due to the high stakes, rational parties may also try to influence the randomness source, even if it is costly to them in the case that they fail.

A reliable source of publicly-verifiable randomness is useful as a basis for many cryptographic primitives. For instance, the design of the \texttt{SHA256} hash function specifies operations such as $y\leftarrow(x \mathsf{\ rotr\ } 6) \oplus (x \mathsf{\ rotr\ } 11) \oplus (x \mathsf{\ rotr\ } 25)$, where $\mathsf{rotr}$ is a circular shift right instruction. %\ariel{Think you can start with the moreover part and skip sentence about operations. Don't see how it's relevant}
Moreover, \texttt{SHA256} specifies some particular $2048$ bits for round constants, as well as some particular $256$ bits for an initial state. To increase the public confidence in the scheme, these $2048+256=2304$ bits were derived by invoking simple functions on a series of small prime numbers, since the \texttt{SHA256} designers claim that other choices for these $2304$ bits would be just as good. Still, there exists the possibility that \texttt{SHA256} has a backdoor (cf. \cite{EPRINT:BLN15}) that was conjured by first picking the $2304$ bits for the round constants and the initial state, and then computing constants for the \texttt{SHA256} algorithm (such as $6,11,25$ above) that enable the backdoor. It is of course possible to outsource the selection of these $2304$ bits to a trusted party, but the question of how to publicly verify that the \texttt{SHA256} designers and the trusted party did not collude remains unanswered. If the \texttt{SHA256} designers could publish the algorithm and define the $2304$ bits to be the unpreditable bits that a public beacon has yet to produce, the confidence in \texttt{SHA256} might be even higher.  %circular logic: bitcoin relies on sha256

To take a more rigorous example, one may consider a provably secure (albeit less efficient) algebraic hash function of the form $H(x,y)=g^x h^y$. Finding a collision $H(x_1,y_1)=H(x_2,y_2)$ is equivalent to breaking a discrete-log hardness assumption (cf. \cite{MAN:Mironov} and \cite[Lecture~4]{MAN:W09}). For this to hold, $g$ and $h$ should be selected as random elements, i.e., without knowing a number $n$ such that $g^n=h$. In groups that are of interest it is easy to sample random elements by tossing public coins (see for example \cite[Section~3.3.4]{MAN:BB02}), and hence the unpreditable bits that a public beacon produces can be used to define the element $h$.

On the other hand, a public beacon is unhelpful in the case of cryptographic primitives that rely on a {\em structured} common reference string (CRS), because private randomness is needed in order to produce the secrets that the structured CRS is comprised of (cf. \cite{C:GO07,EPRINT:BFS16}). Thus, for a protocol that requires a CRS of the form, say, $s=(g^x,g^y,g^{xy})$, a trusted party can use her private randomness to sample $x,y$ and compute $s$, while a public beacon can sample $g^x,g^y$ but will not be able to output $s$ (if computational Diffie-Hellman \cite{DiffieHellman76,ANTS:Boneh98} is intractable). To give some concrete examples, the Zerocash cryptocurrency \cite{SP:Z14} and other constructions that are based on linear PCPs~\cite{CCC:IshaiKO07,TCC:BCIPO13} cannot utilize a public beacon. Many other NIZK constructions require a structured CRS, for example \cite{EC:GS08}. Protocols that sample a structured CRS (e.g., \cite{SP:BCGTV15}) are not publicly verifiable, and may become increasingly susceptible to aborts (a.k.a. denial of service) in conjunction with the number of parties that execute the sampling algorithm. Additionally, let us remark that when the stakes are high the standard definitions do not necessarily capture a real-world setting. For instance, consider $n$ parties that execute $n$ or fewer invocations of a protocol with identifiable aborts~\cite{C:IOZ14}, until an invocation terminates with no aborts by any of the remaining parties. If at least one party is honest then she always contributes her input (i.e., private randomness) to the structured CRS and then destroys her input. However, a single party who remains alone in the final invocation is then more likely to become corrupt (e.g., she could be pressured by the other corrupt parties, or she might be willing to be bribed).

A good randomness source can be useful for an individual user too. Since it is difficult to operate in a completely secure environment, the personal computer of a user could be infected with malware. If the user wishes to run a sensitive process that depends on randomness, she may be concerned about the possibility that malware would interact with this process and feed it non-random bits. Hence, if the process can run a verification algorithm that ensures that the bits originated from a reliable randomness source, the prospects of a successful attack are diminished. % (cf. \cite{hmmmmmmmmm}).  % seed for ``cleaning'' a local source of entropy that might be under the influence of an attacker.

One possible method for obtaining unpreditable random bits is to use financial data \cite{ClarkH10}. For example, we can derive the output from the least significant bits (LSBs) of the end-of-day price of some particular assets that are traded on a stock exchange. However, it is problematic to formalize the assumptions and quantify the security that this protocol achieves. If the stakes are high enough then it might for example be possible to bribe a person who controls the display of the stock prices, so that she would tweak the LSBs of the end-of-day price of the particular assets. Furthermore, attempting to amplify the security of this protocol is nontrivial: if we wish to derive the beacon output from two distant stock exchanges (say NYSE and SEHK), then a person who controls the display in one stock exchange can wait and see the prices of the other exchange, and then modify the LSBs in the stock exchange that she controls accordingly.

\smallskip
An alternative approach is the NIST beacon \cite{WWW:NISTBEACON,RSA:P14}, though it relies on a trusted party.

\smallskip
In this work we explore whether the Bitcoin model \cite{MAN:n08,EC:GKL15} is a reliable source of public randomness. Clearly, in the absence of an attacker, the entropy that is used to generate fresh secret keys and the assumption that the {\em Proof of Work} hash resembles a random oracle imply that plenty of random bits can be extracted from the Bitcoin ledger. In the presence of an attacker, our results can be summarized as follows:
\begin{enumerate}
\item If the attacker has $p$ fraction of the Bitcoin mining power and she is able to generate Bitcoin blocks at a loss, indefinitely, then it is impossible to extract from the Bitcoin ledger a single bit whose statistical distance from random is smaller than $\frac{1}{12} \cdot p$. %$C_0 \cdot p$, where $C_0 = \frac{1}{12}$
%$C_0 \geq \frac{1}{12}$
%is a constant.
%\footnote{Inspection of the proof of Theorem \ref{thm:mainLB} shows
%that if we only want the result to hold for large enough $p<1$, we can improve and get $C_0$ to be any constant smaller than $1/4$. See Remark \ref{rem:tightness} for details.}
\item Assuming that the attacker has a limited budget that prevents her from generating blocks at a loss for too long, it is possible to extract bits from the Bitcoin ledger with statistical distance from random that is arbitrarily close to $0$.
\end{enumerate}
%More precisely, for a large enough $p<1$ we can improve and get $C_0$ to be any constant smaller than $\nicefrac{1}{8}$ (cf. Remark~\ref{rem:tightness}).  
For the sake of comparison, consider a weak adversary who never attempts to fork the blockchain. Thus, in this setting, our negative result shows that any beacon protocol is susceptible to bias that is at most $6$ times smaller than the bias of a trivial protocol that extracts the output from a single predetermined block. %, where $c\in(4,6]$ depends on the power of the adversary.
This is because the trivial protocol can simply hope that the honest miners created the predetermined block, and thereby achieve $\frac{1}{2}p$ or smaller statistical distance from random (an adversary who ``resets'' once achieves $\abs{\frac{1}{2}(1-p)+\frac{3}{4}p-\frac{1}{2}}=\frac{1}{4}p$).
Indeed, Santha-Vazirani sources \cite{FOCS:SanVaz84,BLOG:RVW12} exhibit a similar behavior, and we make use of generalized Santha-Vazirani sources \cite{ICALP:BEG15} to derive our lower bound.
%\ariel{Clearer?: `outputs the first bit of a single predetermined block'. In any case, might want to use `predetermined' in your phrasing} .  
%To be more precise, the phrase ``operate at a loss'' can be thought of as saying that the attacker is able to discard her solved blocks without repercussions.
%\ariel{This phrasing doesn't make sense. If we are being optimistic and assuming honest miners created block bias is 0. Maybe `..because the trivial protocol, outputs a uniform bit given that the block is created by honest miners which happens with probability 1-p; and thus the total bias is at most $\frac{1}{2}$ }

In case the attacker has less than $\nicefrac{1}{2}$ of the mining power and her budget is small enough, our positive results require $n=O(\frac{1}{\eps^2}\log^2\frac{1}{\eps})$ Bitcoin blocks to obtain a bit whose statistical distance from random is at most $\eps$. For adversaries with large budgets, our protocols may require an even larger $n$. To compare, the lower bound \cite{STOC:Cleve86} and upper bound \cite{TCC:MoranNS09} on a two-party coin-flipping protocol have $\Theta(\frac{1}{\eps})$ rounds for output with $\eps$ bias. Similar upper bounds exist for multiparty coin-flipping protocols in the case that at most $\nicefrac{2}{3}$ of the parties are corrupt \cite{CRYPTO:BeimelOO10,STOC:HaitnerT14}, and the best known protocols when more than $\nicefrac{2}{3}$ of the parties are corrupt require $\Omega(\frac{1}{\eps^2})$ rounds for output with $\eps$ bias \cite{STOC:Cleve86,CRYPTO:BeimelOO10}. The comparison between blocks and rounds has some sense to it, because both a Bitcoin block and a round of interaction in a multiparty protocol require propagation of messages on a network. However, in other regards these two notions are incomparable. On the one hand, a Bitcoin block requires a significantly longer time ($10$ minutes on average) than a typical round of interaction. On the other hand, fetching the Bitcoin block can be done by a party who does not even participate in the protocol.

The last remark alludes to the main advantage of Bitcoin-based randomness extraction protocols: unpredictable yet publicly verifiable random bits. In fact, the public verifiability aspect of the beacon has two advantages that are closely related. First, it enables incorruptible protocols by removing any element of trust in the parties who initiated the system. Second, it allows everyone, including new parties who have just begun their participation, to be in agreement on the output of the beacon. Indeed, our results in Section~\ref{sec:upperbound} show that all parties reach consensus on the output that the beacon generates.

\subsection{Related works}
The idea of using Bitcoin as a public randomness source has been explored in \cite{C:AD15} and \cite{EPRINT:BCG15}. The work of \cite{EPRINT:PW16} presents negative results with regard to Bitcoin-based randomness extraction in the presence of unrestricted and budget-restricted adversaries. However, \cite{EPRINT:PW16} considers protocols that extract the output from only a single Bitcoin block, unlike our positive results in Section~\ref{sec:upperbound}.
Also, our lower bound in Section~\ref{sec:lowerbound} applies to protocols
that can use an arbitrary function of all blocks witnessed as the outputted randomness, 
and considers a weak adversary who may not fork the chain, 
so in this sense our lower bound is more general than the lower bound of \cite{EPRINT:PW16}.

\subsection{Organization of the paper}
The contributions of this work are organized as follows. In Section~\ref{sec:pre} we provide some well-known definitions and tools. In Section~\ref{sec:lowerbound} we prove that no protocol can achieve an arbitrarily small bias when the adversary has an infinite budget. In Section~\ref{sec:upperbound} we analyse beacon protocols that defeat a budget-restricted adversary, first in a simplified model and then in a model that captures Bitcoin. In Section~\ref{sec:profitmargin} we discuss the rationale for the assumptions that need to be made with regard to budget-restricted adversaries. In Section~\ref{sec:hybrid} we describe a protocol that combines a public beacon with reliance of trusted parties. In Section~\ref{sec:practical} we discuss some real-world considerations that a Bitcoin-based beacon should take into account. 

\ignore{
\hs{we need to mention that NIST beacon cannot be fully trusted given the fact that certain services e.g., ECDSA has been blamed that NIST has backdoors built in. In this proposal we only assume the hash function provided by NIST has certain property.}

NIST beacon: nice properties and cool applications.

Drawbacks of NIST beacon.

a distributed beacon is highly needed.  

potential constructions.

here is a contraction based on Bitcoin system; 
}

%\newpage

% !TEX root = mainBeacon.tex

\section{Preliminaries}\label{sec:pre}

%\subsection{Hash-based proof of work}

%\subsection{Bitcoin system}

%transactions

We denote $[n]\triangleq\{1,2,\ldots,n\}$.

Let $\Omega$ denote some finite domain, frequently of the form $\Omega=\Sigma^n$ for $\Sigma=[d]$.

For a subset $S\subseteq \Omega$,
we denote $\mu(S)\triangleq |S|/|\Omega|$.

We denote by $\dtv(X,Y)\triangleq\frac{1}{2}\Sigma_{a\in\Omega}\abs{X(a)-Y(a)}$ the
statistical distance between two distributions $X,Y$.
Note that in case $U$ is the uniform distribution on $\Omega=\{0,1\}$ and $X$ also has support $\{0,1\}$, it holds
that $\abs{X(0)-\frac{1}{2}}=\abs{(1-X(0))-\frac{1}{2}}=\abs{X(1)-\frac{1}{2}}=\dtv(X,U)$.

%As is done frequently in theoretical computer science (without even noticing),
%we will sometimes use the term distribution
%to refer to a random variable distributed according to that distribution.

%\subsection{Symbol-fixing sources and extractors}

\begin{definition}[non-oblivious symbol-fixing sources]
A distribution $X$ over $\Sigma^n$ is an $(n,k,\Sigma)$ non-oblivious symbol-fixing source if there exists a subset $T=\{i_1,i_2,\ldots,i_k\}\subseteq [n]$ and a function $f:\Sigma^k\rightarrow\Sigma^{n-k}$ such that $(X_{i_1},X_{i_2},\ldots,X_{i_k})$ is uniformly distributed over $\Sigma^k$, and $(X_{j_1},X_{j_2},\ldots,X_{j_{n-k}})=f(X_{i_1},X_{i_2},\ldots,X_{i_k})$ where $\{j_1,j_2,\ldots,j_{n-k}\}=[n]\setminus T$.
\end{definition}

In the special case of $\Sigma=\{0,1\}$, we say that $X$ is an $(n,k)$ non-oblivious bit-fixing source.

\begin{definition}[extractor for symbol-fixing sources]
Let $S$ be a family of $(n,k,\Sigma)$ symbol-fixing sources. An $\eps$-extractor for $S$ is a function $\extr:\Sigma^n\rightarrow\Sigma$ such that $\forall X\in S:\dtv(\extr(X),U)\leq\eps$, where $U$ is the uniform distribution on $\Sigma$.
\end{definition}

The following lemma gives sufficent conditions under which the majority function can act as a randomness extractor.
In essence, the lemma shows that the difference between the amount of 0s and 1s among $n$ uniform random bits is concentrated at around $\sqrt{n}$, similarly to the expected value of the absolute displacement of a simple random walk.

\begin{lemma}\label{lemma:majextr}
Let
$f(x_1,x_2,\ldots,x_n)=
\left\{\begin{array}{ll}
    1 & \mbox{$\sum^n_{i=1} x_i\geq\frac{n}{2}$}\\
    0 & \mbox{ \rm{otherwise} }
    \end{array}\right.$
be the majority function, and let $\ell\triangleq\bigl\lfloor\eps\frac{\pi}{\euler}\sqrt{n-\sqrt{n}}\bigr\rfloor$. If $n$ is odd and $\eps\leq\frac{\euler}{\pi}$ then $f$ is an $\frac{\eps}{2}$-extractor for an $(n,n-\ell+1)$ non-oblivious bit-fixing source.
\end{lemma}
\begin{proof}
By Stirling's approximation, $\frac{1}{2^n}{n \choose n/2}\leq\frac{1}{2^n}\frac{\euler\sqrt{n}(\frac{n}{\euler})^n}{(\sqrt{\pi n}(\frac{n}{2\euler})^{n/2})^2}=\frac{\euler}{\pi}\frac{1}{\sqrt{n}}$. Therefore, for every $k\geq 0$, it holds that $\Pr(k=\sum^n_{i=1} x_i)\leq\Pr(\floor{\frac{n}{2}}=\sum^n_{i=1} x_i)\leq\frac{\euler}{\pi}\frac{1}{\sqrt{n}}$. Let $E$ denote the event that $f$ remains undetermined after $n-\ell$ variables were chosen randomly. Hence,
\begin{eqnarray*}
\pr(E)&=&\sum\limits^{\ell}_{k=1}\Pr\left(\sum\nolimits^{n-\ell}_{i=1} x_i=\Bigl\lceil\frac{n}{2}\Bigr\rceil-k\right)\leq\sum\limits^{\ell}_{k=1}\frac{\euler}{\pi}\frac{1}{\sqrt{n-\ell}}\overbrace{\leq}^{\eps\hspace{2pt}\leq\hspace{1pt}\frac{\euler}{\pi}}\ell\cdot\frac{\euler}{\pi}\frac{1}{\sqrt{n-\sqrt{n}}} \\
&\leq&\Bigl\lfloor\eps\frac{\pi}{\euler}\sqrt{n-\sqrt{n}}\Bigr\rfloor\cdot \frac{\euler}{\pi}\frac{1}{\sqrt{n-\sqrt{n}}}\leq\eps.
\end{eqnarray*}
Assume that $\ell$ is even. Since $n$ is odd, it holds that $\pr(f(x_1,\ldots,x_n)=1|\lnot E)=\pr(\sum\nolimits^{n-\ell}_{i=1} x_i\in\{\ceil{\frac{n}{2}},\ceil{\frac{n}{2}}+1,\ldots,n-\ell\})$ and $\pr(f(x_1,\ldots,x_n)=0|\lnot E)=\pr(\sum\nolimits^{n-\ell}_{i=1} x_i\in\{0,1,\ldots,\floor{\frac{n}{2}}-\ell\})$. Since ${{n-\ell} \choose {\ceil{\frac{n}{2}}+k}}={{n-\ell} \choose {n-\ell-\ceil{\frac{n}{2}}-k}}={{n-\ell} \choose {\floor{\frac{n}{2}}-\ell-k}}$ for $k=0,1,\ldots,{\floor{\frac{n}{2}}-\ell}$, we have that $\pr(f=1|\lnot E)=\pr(f=0|\lnot E)=\nicefrac{1}{2}$. This implies that $f$ is an $\frac{\eps}{2}$-extractor for an $(n,n-\ell)$ non-oblivious bit-fixing source, because
\begin{eqnarray*}
\dtv(f(X),U) &=& \Bigl\lvert \pr\bigl(f(x_1,\ldots,x_n)=0\bigr)-\frac{1}{2} \Bigr\rvert\\
%&=& \Bigl\lvert \pr\bigl(\{f=0\}\land E\bigr) + \pr\bigl(\{f=0\}\land\lnot E\bigr)-\frac{1}{2} \Bigr\rvert\\
& = & \Bigl\lvert \pr\bigl(\{f=0\}\cap E\bigr) + \pr\bigl(\{f=0\} | \lnot E \bigr)\cdot\pr(\lnot E)-\frac{1}{2} \Bigr\rvert\\
& = & \Bigl\lvert \pr\bigl(\{f=0\}\cap E\bigr) + \frac{1}{2}\cdot\bigl(1-\pr(E)\bigr)-\frac{1}{2} \Bigr\rvert\\
& = &  \Bigl\lvert \pr\bigl(\{f=0\}\cap E\bigr) - \frac{1}{2}\pr(E) \Bigr\rvert
 \hspace{2pt}\leq\hspace{2pt}  \frac{1}{2}\pr(E) \hspace{2pt}\leq\hspace{2pt} \frac{1}{2}\eps.
\end{eqnarray*}
%$$\dtv(f(X),U)\leq\Bigl\lvert (1-\pr(E))\cdot\pr(f(x_1,\ldots,x_n)=0|\lnot E)-\frac{1}{2} \Bigr\rvert=\frac{1}{2}\pr(E)\leq\frac{1}{2}\eps.$$
If $\ell$ is odd, the result follows by replacing $\ell$ with $\ell'=\ell-1$.
\end{proof}

\def\B{\{0,1\}}
\def\C{X_1=x_1,\ldots,X_{i=1}=x_{i-1}}
\newcommand{\prandd}{$p$-2-easy source}
\newcommand{\prand}[1]{$p$-2-easy $d-#1$ source}
\newcommand{\prandfamily}[1]{$p$-2-easy $d-#1$ source family}
\newcommand{\dn}{\ensuremath{[d]^{n}}\xspace}
\newcommand{\dd}[1]{\ensuremath{[d]^{#1}}\xspace}
\renewcommand{\d}{\ensuremath{[d]}\xspace}
\newcommand{\pre}[1]{#1'}

\section{Adversaries with an infinite budget}\label{sec:lowerbound}
In this section we prove our lower bound. Our proof is inspired by 
a lower bound on extraction from generalized Santha-Vazirani (SV) sources from Appendix B of  Beigi, Etesami and Gohari~\cite{ICALP:BEG15}.
We mention that this proof is based, in turn, on an elegant proof of 
% Omer 
 Reingold, 
%Salil
 Vadhan and 
%Avi 
 Wigderson \cite{BLOG:RVW12} simplifying the original lower bound of \cite{FOCS:SanVaz84} for randomness extraction from SV sources.

\ignore{
%\paragraph{Notation}
For a subset $S\subseteq \Omega$ 
of some finite domain $\Omega$,
we denote $\mu(S)\triangleq |S|/|\Omega|$.
Frequently, $\Omega$ will be $\dn$.
As is done frequently in theoretical computer science (without even noticing),
we will sometimes use the term distribution,
to refer to a random variable distributed according to that distribution.
\ariel{perhaps move notation paragraph to preliminaries section}
}

\newcommand{\psource}{$p$-resettable source\xspace}
\newcommand{\preset}{$p$-resettable\xspace}

\newcommand{\dist}{\ensuremath{X}\xspace}
\newcommand{\prex}{\ensuremath{x_1,\ldots,x_{i-1}}}
\newcommand{\preX}{\ensuremath{X_1,\ldots,X_{i-1}}}
\newcommand{\predn}{\ensuremath{[d]^{-1}}}
\newcommand{\preevent}{\ensuremath{X_1=x_1,\ldots,X_{i-1}=x_{i-1}}}

The purpose of the next definition is to formally model the distribution of blocks
generated when an adversary has control of a $p$-th fraction of the mining power.
We define a relatively weak adversary \adv.
Specifically, \adv does not have the power to fork the network and try to create alternate chains.
The only thing \adv is able to do is to try to mine the next block himself, and if he succeeds before the honest miners, he can decide whether to publish this block or let the honest miners publish their version of the next block.
Furthermore, \adv must make this decision \emph{before} seeing the honest miners' version of the 
next block.
Note that one may regard \adv to be an adversary with an infinite budget, under the interpretation that producing each block has a cost and \adv can discard block rewards without repercussions.
Let us emphasize that since we use this definition for our randomness extraction \emph{impossibility result},
the weak adversary model makes the result \emph{stronger}.
\ariel{This sentence is confusing in the context of this section -  Would delete it. You're forcing the reader to juggle to many things that are not relevant for the lower bound.  A better place to mention it is in intro or later when you discuss the forkless model. Update by iddo: added "later"}Indeed, if we assume that all adversaries have an infinite budget, \adv is weaker than the adversary $\mc{A}_1$ of the $\nofork$ model (cf. Figure~\ref{fig:forkless}), yet $\mc{A}_1$ is the weakest adversary that we will consider later in our positive results.
We proceed to the formal definition.
\begin{definition}[\psource]\label{dfn:p-source}
 Fix integers $d$ and $n$, and $0<p\leq 1$.
 A distribution \dist on \dn is a \emph{\psource}, or simply \preset,
 if it can be sampled symbol by symbol by a randomized algorithm \adv via
 the following process.
 For any $\prex \in \dd{i-1}$, given that
 \preevent,
 $X_i$ is sampled as follows.
 \begin{enumerate}
  \item A uniform element $a\in \d$ is chosen.
  \item With probability $1-p$, $X_i$ is set to $a$.
  Otherwise, $a$  is given to algorithm \adv . \adv 
  can now choose, given access to $a$ and \prex, either to set $X_i =a$, or discard $a$ and then choose a new uniform element
  $b\in \d$ and set $X_i=b$.
 \end{enumerate}
We emphasize that \adv must choose whether to `reset' $X_i$, i.e. to discard $a$,  \emph{before} sampling $b$.
\end{definition}
Note that using the definition above also for $n=1$, we have that
if for every $i\in [n]$, and every $\prex \in \dd{i-1}$
$(X_i|\preevent)$ is \preset, then \dist is \preset.

\paragraph{Proof idea:}
We sketch the idea of the proof, relating it to the proofs of \cite{BLOG:RVW12} and \cite{ICALP:BEG15}.
\cite{BLOG:RVW12} show an extraction lower bound for the family of `slightly imbalanced sources'.
These are distributions on \dd{n} with the property that the ratio of the probabilities
given to any two strings in \dd{n} is close to one.
Most of our work will be to show that any slightly imbalanced source is in fact a
\psource.
This is a similar strategy to \cite{ICALP:BEG15}, that `embed' slightly imbalanced sources into generalized SV sources for the purpose of their lower bound.
We begin with the definition of a $p$-perturbed distribution on \d,
which will be useful for this purpose.

\newcommand{\pdist}{$p$-perturbed}
\newcommand{\pertrubed}[1]{$#1$-perturbed}
\begin{definition}\label{dfn:pdist}
A distribution \dist on \d is called \emph{\pdist},
if for any element $a\in \d$,
\[(1-p)/d \leq \pr(\dist =a )\leq (1+p)/d.\]
\end{definition}
\begin{lemma}\label{lem:pertrubed_is_source}
Let \dist be a distribution on \d, and fix any $0<p\leq 1$.
If \dist is \pertrubed{p/2} then $X$ is a \psource.
\end{lemma}
\begin{proof}
 Fix \dist that is \pertrubed{p/2}.
 For $a\in \d$, we define a quantity $u_a$ to measure ``how much probability $a$ is given beyond required minimum''.
 Formally, we define 
 \[u_a\triangleq \pr(\dist = a)\cdot d - (1-p/2).\]
 Note that, as \dist is \pertrubed{p/2}, $0\leq u_a\leq p$ for all $a\in \d$.
 %Define $\mu \triangleq \expect_{a\in \d}[u_a]$, satisfying $\mu \leq p$.
 
%Now define, for $i\in \d$, $\gamma_i\triangleq u_i 

Consider the following sampling procedure.
\begin{enumerate}
 \item Sample $a\in \d$ uniformly.
 \item With probability $1-p+u_a$, output $a$.
 Otherwise, output a uniform $b\in \d$.
\end{enumerate}
We claim the distrbution $Y$ sampled by this procedure is a \psource.
This is because the ``resetting'', i.e. outputting $b$ rather than $a$,
always happens with probability at most $p$, even after conditioning on the value of $a$.
We denote by $\eta$ the probability that resetting occurred; i.e.,  the overall probability that the procedure outputs $b$ rather than $a$, \emph{without} conditioning on the value of $a$.
We have
\[\mu= \expect_{a\in \d}[p-u_a] = p/2,\]
as $\expect_{a\in \d} [u_a] = p/2$.

Fix any $c\in [d]$. We will show that 
$\pr(Y=c)= \pr(\dist=c)$; this implies $Y\equiv \dist$ which means
\dist is a \psource.
The event $Y=c$ is a union of the following two \emph{disjoint} events $A$ and $B$:
\begin{itemize}
 \item $A$: $c$ was output as $a$ in the first stage of the procedure.
 We have 
 \[\pr(A)=1/d\cdot(1-p+u_c).\]
 \item $B$: $c$ was output as $b$ in the resetting stage.
 $\pr(B)$ is the product of the probability $\eta=p/2$ that resetting occurred
 times $1/d$, as given that we are resetting, the output is uniform.
\end{itemize}

Thus, for any $c\in \d$, we have
\[\pr(Y=c) = (1-p+u_c)/d  +(p/2)/d = (1-p/2+u_c)/d = \pr(X=c). \]
Therefore, \dist is a \psource.
\end{proof}

\begin{claim}\label{clm:inequality}
 Fix $0< q\leq 1/3$.
 We have
 \[1-2q\leq \frac{1-q}{1+q} \leq \frac{1+q}{1-q} \leq 1+3q.\]
\end{claim}

\begin{theorem}\label{thm:mainLB}
For any function $E:\dn\to \B$ and any $0<p\leq 1$, there is a \psource
\dist such that $E(\dist)$ has bias at least $p/12$.
\end{theorem}

\begin{proof}
 Fix any function $E:\dn\to \B$, and any $0<p\leq 1$.
 Assume w.l.o.g. that $\mu(E^{-1}(0)) \geq 1/2$,
 and fix a set $S\subseteq E^{-1}(0)$ with $\mu(S)=1/2$ 
 (Here we assumed for simplicity $d$ is even. Otherwise,
 the proof can be altered).
 Define a distribution \dist on $\dn$ as follows.
 \begin{itemize}
  \item Any $x\in S$ has probability $(1+p/6)/d^n$
  \item Any $x\notin S$ has probability $(1-p/6)/d^n$
  \end{itemize}
  
  We have $\pr(E(X) = 0)\geq 1/2 + p/12$.
It is left to show that \dist is a \psource.
   Note that for any set $T\subseteq \dn$,
 \[(1-p/6)\cdot \mu(T) \leq \pr(X\in T) \leq  (1+p/6)\cdot \mu(T).\]
  
  Fix $i\in [n]$, $\prex\in \dd{i-1}$, and $a\in \d$.
  Denote
  \[\eta_a \triangleq \pr(X_i =a | \preevent) = \frac{\pr(X_1=x_1,\ldots,X_{i-1}=x_{i=1},X_i =a)}{\pr(\preevent)}\]
  
  We know that
 \[  (1-p/6)\cdot d^{- n+(i+1)} \leq \pr(\preevent)\leq(1+p/6)\cdot d^{-n+(i+1)},\] 
 \[  (1-p/6)\cdot d^{-n+i} \leq \pr(X_1=x_1,\ldots,X_{i}=x_{i}) \leq(1+p/6)\cdot d^{-n+i}. \]
  
  Using Claim \ref{clm:inequality} with $q=p/6$, we have
  \[ (1-p/3)/d\leq \frac{1-p/6}{d(1+p/6)}\leq \eta_a \leq \frac{1+p/6}{d(1-p/6)}\leq (1+ p/2)/d\]
  Thus, for any $i\in [n]$, $\prex\in \dd{i-1}$,
  $(X_i|\preevent)$ is \pertrubed{p/2}, and therefore \preset.
  It follows that \dist is a \psource.
  \end{proof}
\begin{remark}[Extractors with unbounded input length]
 One may wonder whether an extractor that does not
 have a fixed input length $n$ can get smaller error.
  However, for any such extractor $E$  
 and any $\eps>0$,  we can take $n$ such
 that the probability that $E$ produces an output after $n$ symbols is
 at least $1-\eps$, and define $E'$ to be the extractor
 that reads $n$ symbols and answers according to $E$,
 or answers $0$ if $E$ hasn't terminated.
 (It seems a reasonable assumption that such $n$ exists if $E$ is to be useful).
 By the lower bound, $E'$ has error at least $p/12$,
 which implies the error of $E$ is at least $p/12 - \eps$.
 As this holds for any $\eps>0$, the error of $E$ must be at least, say, $p/13$.
\end{remark}
\begin{remark}[Efficiency of the adversary]\label{rem:effAdvLB}
 The proof of Theorem \ref{thm:mainLB} shows that there \emph{exists} a \psource
 that will bias the output of $E$. One may wonder, if the ``adversary can be efficient'';
 that is, whether the algorithm \adv from Definition \ref{dfn:p-source} - deciding whether to ``reset or not'', can be efficient.
 Examination of the proof shows that when $E$ is computable in time $\poly(n,d)$;
 we can construct a \psource $X'$ where \adv will operate in time $\poly(n,d)$
 and such that $E(X')$ will have bias at least $p/13$.
 We sketch why this is so.
 Let $X$ be the distribution defined in the beginning of the proof of Theorem \ref{thm:mainLB}.
 We assume for simplcity here that $\pr(E(x)=0) = 1/2$; otherwise, we must start by approximating this probability and this will add an arbitrarily small error term.
 We first observe that when $E$ is efficiently computable, the distrbution $X$ can be 
 efficiently sampled:
 Simply choose random $x\in \dn$, compute $E(x)$, and discard $x$ and resample with the appropriate probability when $E(x)=1$, to give $x$'s with $E(x)=0$ the desired larger probability $(1+p/6)/d^n$.
 In a similar way, we can efficiently sample distributions of the form $(X|\preevent)$, by starting with a random $x\in \dn$ with prefix $(x_1,\ldots,x_{i-1})$, rather than a completely uniform $x\in\dn$.
 Thus, we can approximate the probabilities $\eta_a \triangleq \pr(X_i =a | \preevent)$ appearing in the proof.
 Inspection of the proof of Lemma \ref{lem:pertrubed_is_source} shows these probabilities are all we need to sample $X$ correctly; thus, \adv can efficiently approximate $X$.
\end{remark}

%  \begin{remark}[Tightness of the bound]\label{rem:tightness}
%   Inspection of the proof of Theorem \ref{thm:mainLB} shows that if we only want the result to hold for large enough $p\leq 1$, we can replace the term $p/12$ in Theorem \ref{thm:mainLB} by $p/C$ for any $C>8$.
%  The relevant observation is that the constant $3$ in the statement of Claim \ref{clm:inequality} can be replaced by any constant larger than $2$ for small enough $0<q<1$.
% % 
% \end{remark}

%\end{document}

%\newpage

\def\maxcoins{\mathsf{maxprofits}}
\def\kparam{k}
\def\claimtwo{Lemma\xspace}

\section{Beacon for budget-restricted adversaries}\label{sec:upperbound}
It is entirely reasonable to assume that a Bitcoin miner will incur significant losses if she does not claim the rewards for a large portion (such as half) of the blocks that she produces, and hence she will deplete her budget when she employs a strategy of this kind. Still, the beacon protocol needs to establish its correctness in the face of an adversary who wishes to offset such losses by attempting to accumulate profits during certain periods of the protocol execution. With this in mind, let us provide the following abstract assumption.

\begin{assumption}\label{assume:maxprof} Let $\maxcoins:\mathbb{Q}\times\mathbb{N}\rightarrow\mathbb{Q}$ be a function that is monotonically increasing in both of its arguments. Consider a Bitcoin miner who invested $t$ coins to buy her mining equipment. For any segment of $n$ consecutive blocks of the chain, the amount of coins profit (i.e., revenues minus operating expenses) that this miner will earn as a result of creating some portion of the blocks in this segment is bounded by $\maxcoins(t,n)$.
\end{assumption}
The function $\maxcoins$ bounds the profitability of generating blocks. For now, $\maxcoins$ can be thought of as satisfying the condition $\forall n: \maxcoins(t,n)\leq 2t$. See Section \ref{sec:profitmargin} for a more thorough examination of Assumption~\ref{assume:maxprof}, and in particular Example \ref{example:profits} with regard to why Assumption~\ref{assume:maxprof} is unavoidable.

%\subsection{Adversary that only extends that head of the chain}
\medskip

To explain the main considerations that come into play in a budget-restricted setting, let us first provide a proof in a simplified model $\nofork$, that is specified in Figure~\ref{fig:forkless}. In this model, blockchain forks never occur, and all parties including the adversary create blocks only at the head of the chain. %The model $\nofork$ is specified in Figure~\ref{fig:forkless}.
Note that we associate each block with a formal symbol, and so we use the terms ``block'' and ``symbol'' interchangeably. %In this model, a linked chain of blocks is constructed ad infinitum.

\begin{figure}[htb]
\begin{minipage}{\textwidth}
\begin{framed}
{
\small
\begin{center} Forkless model $\nofork$ \end{center}
The forkless model with parameter $p$ proceeds in turns, as follows:
\begin{enumerate}
\item In each turn, the adversary $\mc{A}_1$ is successful with probability $p$ and unsuccessful with probability~$1-p$.
\item In case $\mc{A}_1$ is successful, she receives a uniform random block from the random oracle and can either discard it or extend the head of the chain with it.
\item In case $\mc{A}_1$ is unsuccessful, the chain is extended with a uniform random block.
\end{enumerate}
}
\end{framed}
\end{minipage}
\caption{The forkless model.}\label{fig:forkless}
\end{figure}

Let us consider a protocol that outputs a single bit. Since we may assume w.l.o.g. that the adversary wishes to bias the output towards $1$, and that each block is associated with a symbol that is represented by bits, we define the following notation.

\begin{definition}\label{def:helpdet}
A ``helpful'' block is a block such that the symbol that is associated with it has $1$ as its least significant bit (LSB). Likewise, a ``detrimental'' block is a block such that the symbol that is associated with it has LSB that equals $0$.
\end{definition}

%\medskip
In the following lemma, the parameter $\delta$ measures how large the profit margin of the adversary can be. For example, one may think of $\delta=\nicefrac{2}{3}$ throughout the analysis. If $p$ is rather small so that $p'=\frac{\nicefrac{p}{2}}{1-\nicefrac{p}{2}}\approx\nicefrac{p}{2}$, the profit margin $z_p=px-y_p$ should be small enough so that $w_p=\frac{1}{\delta}p' x-y_p\approx\frac{1}{\delta}\frac{1}{2}p x -y_p=\frac{3}{4}p x-y_p$ is a negative number. Even in the case that $p$ is quite large, say $p=\nicefrac{1}{5}$, if we substitute exemplary numbers such as $x=50,y_p=9,z_p=\nicefrac{50}{5}-9=1$, we get that $w_p=\nicefrac{3}{2}\cdot\nicefrac{1}{9}\cdot 50-9=\nicefrac{-2}{3}<0$. Thus, in the case of $\delta=\nicefrac{2}{3}$, the parameter $n$ should be large enough so that $p_0=\exp\{-\frac{1}{3}(\frac{1}{\delta}-1)^2 \delta\ell\}=\exp\{-\frac{\ell}{18}\}$ is negligible.

\begin{lemma}\label{claim:upbound1}
%Assume that the \bc hash function is a random oracle.
Consider an adversary $\mc{A}_1$ who operates in the model $\nofork$ with parameter $p$. Assume that $\mc{A}_1$ purchased her $p$ fraction of the mining power with $t_1$ coins, and that she has a reserve of $t_2$ additional coins. Denote $T(i)\triangleq t_2+\maxcoins(t_1,i)$. Assume that each block reward is $x$ coins, and that each attempt by $\mc{A}_1$ to solve the next block costs her at least $y_p$ coins. Denote the maximal profit margin of $\mc{A}_1$ by $z_p\triangleq px-y_p>0$ coins. Denote $p'\triangleq\frac{\nicefrac{p}{2}}{1-\nicefrac{p}{2}}$. Assume that $z_p$ is small enough so that $w_p\triangleq \frac{1}{\delta}p' x-y_p<0$ for some constant $\delta\in(\frac{1}{2},1)$. Assume that $\mc{A}_1$ stops operating if she runs out of coins. Consider the following beacon protocol:
\begin{enumerate}
\item Fetch $n$ consecutive blocks $B_1,B_2,\ldots,B_n$ from the blockchain.
\item Obtain a bit $b_i$ from each $B_i$ by taking the LSB of the symbol that is associated with $B_i$. 
\item Output the bit $\mathsf{majority}(b_1,b_2,\ldots,b_n)$.
\end{enumerate}
Let $\eps>0$ be arbitrary. %Assume that $k$ is large enough so that $B_n$ will not be reversed except with negligible probability.
Denote $\ell\triangleq \floor{\eps\frac{\pi}{\euler}\sqrt{n-\sqrt{n}}}$, $p_0\triangleq e^{-\frac{1}{3}(\frac{1}{\delta}-1)^2 \delta\ell}$. If $n$ is taken to be large enough so that $T(n)+ \delta\frac{1}{p'} \ell w_p<0$, then this protocol outputs a bit whose statistical distance from a uniform random bit is at most $\eps+p_0+\frac{\euler}{\pi}\frac{1}{\sqrt{n-\ell}}$.
\end{lemma}
\begin{proof}
The adversary $\mc{A}_1$ may opt to increase her $t_2$ coins reserve by operating as an honest miner in some locations, even though that entails that $\mc{A}_1$ would then have less opportunities to influence the output of the beacon in the remaining locations. Assume w.l.o.g. that $\mc{A}_1$ wishes to bias the output of the beacon towards $1$. Since $\mc{A}_1$ only extends the head of the chain, the optimal strategy that she can deploy must be of a form in which $\mc{A}_1$ alternates between the following two modes. In the first mode, $\mc{A}_1$ operates as an honest miner, meaning that she publishes every block $B_i$ that she creates, regardless of whether $B_i$ is a ``helpful'' or ``detrimental'' block. In the second mode, $\mc{A}_1$ inspects the LSB of each block $B_i$ that she successfully created, and publishes $B_i$ only if it is a ``helpful'' block, otherwise she discards $B_i$. Suppose for the moment that $\mc{A}_1$ never switches to the first mode after she discards a ``detrimental'' block but before the next turn of $\nofork$. The strategy of $\mc{A}_1$ can thus be viewed as adaptively selecting $d\leq n$ locations in which she plays as an honest miner, and in the other $n-d$ locations $\mc{A}_1$ discards all the ``detrimental'' blocks that she creates. Hence, while operating in the second mode, $\mc{A}_1$ will successfully create the block and collect the reward of each location with probability $\frac{1}{2}p+(\frac{1}{2}p)^2+(\frac{1}{2}p)^3+\ldots=\frac{\nicefrac{p}{2}}{1-\nicefrac{p}{2}}=p'$. To decide the outcome of $\ell$ locations, $\mc{A}_1$ needs $\ell$ successes that occur with probability $p'$ each, because a failed trial does not deny the location from having a uniform random symbol. Let $Y=Y(\ell,p')$ be a random variable with negative binomial distribution that counts the number of trials until $\ell$ successes. By using a tail inequality \cite{WWW:browndg,books:daglib0025902} for the \ariel{add `negative'?}binomial distribution $B(\delta\frac{1}{p'}\ell,p')$, we obtain $$\pr(Y<\delta\cdot \expect[Y])=\pr(Y<\delta\frac{1}{p'}\ell)=\pr(B(\delta\frac{1}{p'}\ell,p') > \ell)\leq e^{-\frac{1}{3}(\frac{1}{\delta}-1)^2 \delta\ell}=p_0.$$ If the event $E_2\triangleq\{Y\geq\delta\frac{1}{p'}\ell\}$ occurs then the gains and losses of $\mc{A}_1$ until she influences the majority function $\ell$ times will sum up to $S=\ell x + \delta\frac{1}{p'}\ell(-y_p)=\delta\frac{1}{p'}\ell w_p$ coins or less. Notice that $\mc{A}_1$'s average cost per trial is in fact greater than $y_p$, because each trial may involve several turns.
Also note that if the success probability of $\mc{A}_1$ was $p$ rather than $p'$ then the sum $S$ would be a multiple of $\frac{1}{\delta}px-y_p>z_p$ and hence positive\ariel{not following this last line}. Let us refer to the blocks that $\mc{A}_1$ created while operating in the second mode as {\em marked} blocks, and let $E_3\triangleq\{\mc{A}_1$ contributed at most $\ell$ marked blocks among $B_1,B_2,\ldots,B_n\}$. According to Assumption~\ref{assume:maxprof}, the total amount of coins that $\mc{A}_1$ may earn while playing honestly in the $d$ locations is bounded by $\maxcoins(t_1,d)\leq\maxcoins(t_1,n)$, and hence $T(n)$ represents the maximal budget that $\mc{A}_1$ has for playing in the $n-d$ locations. Since we assume that $T(n)+S<0$ and $\mc{A}_1$ does not go under budget, we have that $E_2\subseteq E_3$ holds, because the occurrence of $E_2$ implies that $\mc{A}_1$ stops operating after she played in at most $\ell$ out of the $n-d$ locations.
Note that if $\delta\frac{1}{p'}\ell > n-d$ then the occurrence of $E_2$ implies that $\mc{A}_1$ failed to reach $\ell$ successes irrespective of her $T(n)$ budget, meaning that $E_2\subseteq E_3$ in this case too.

%, and therefore she has at most $\eps$ probability to bias the output in this case too.
%Moreover, if $E_2$ occurs then $\delta\frac{1}{p'}\ell \leq n-d$ ought to hold, as otherwise some of the blocks that $\mc{A}_1$ solves will fall outside of the $n$-blocks interval that the beacon protocol defines, but let us be lenient towards $\mc{A}_1$ and neglect this requirement.
%Since the majority extractor of Lemma \ref{lemma:majextr} is an $\frac{\eps}{2}$-extractor for $(n,n-\ell)$ non-oblivious bit-fixing source
Let $E_4\triangleq\{$the amount of marked blocks that $\mc{A}_1$ contributed among $B_1,B_2,\ldots,B_n$ is at least as large as the difference between ``helpful'' and ``detrimental'' blocks among the unmarked blocks of $B_1,B_2,\ldots,B_n\}$. Thus, following Lemma \ref{lemma:majextr}, we have that $\pr(E_4|E_3)\leq \eps$.

Let $E_1\triangleq \{\mathsf{majority}(b_1,b_2,\ldots,b_n)=1\}$.
Since a tie among $n-\ell$ uniform random bits occurs with probability $\frac{\euler}{\pi}\frac{1}{\sqrt{n-\ell}}$ at the most (cf. Lemma \ref{lemma:majextr}), we have $\pr(E_1|\lnot E_4 \cap E_3)\leq\frac{1}{2}+\frac{\euler}{\pi}\frac{1}{\sqrt{n-\ell}}$. \ariel{Not following 100 percent.
Does the event $E_2$ mean adversary created at most $\ell$ blocks out of the $n$?}Hence, we obtain

%Hence, for $E_4=\{\mc{A}_1 \ \text{succeeds}\}$, we obtain that
\begin{eqnarray*}
\pr(E_1) &\leq& \pr(E_1 \cap E_3) + \pr(E_1 \cap \lnot E_3) \leq \pr(E_1 \cap E_3) + \pr( \lnot E_3)\\ 
&\leq& \pr(E_1 \cap E_3\cap E_4) + \pr(E_1 \cap E_3\cap \lnot E_4) +\pr( \lnot E_2)\\ 
&\leq& \pr(E_4|E_3) + \pr(E_1|\lnot E_4 \cap E_3) + p_0\\
& \leq & \eps +\bigl(\frac{1}{2}+ \frac{\euler}{\pi}\frac{1}{\sqrt{n-\ell}}\bigr) + p_0.
\end{eqnarray*}
Now, if $\mc{A}_1$ is allowed to switch to the first mode after she discards a ``detrimental'' block but before the next turn, we can still count such a trial as a loss of $y_p$ coins. This means that in case $\mc{A}_1$ succeeds to create the next block it will be accounted for as playing honestly in one of the $d$ locations, and therefore the total number of trials can exceed $n-d$ (implying that the success probability in a trial is less than $p'$). Note that $\mc{A}_1$ should never switch from the first mode to the second mode before she succeeds to create a block, since the only purpose of operating in the first mode is to avoid going under budget. Therefore, $\pr(E_1)\leq \eps + \frac{1}{2}+ \frac{\euler}{\pi}\frac{1}{\sqrt{n-\ell}} + p_0$ holds even if $\mc{A}_1$ may adaptively switch between modes in the course of a trial.

Overall, the statistical distance between the output bit and a uniform random bit is at most 
$$\left\lvert\pr(E_1)-\frac{1}{2}\right\rvert\leq \left\lvert\eps+\frac{1}{2}+ \frac{\euler}{\pi}\frac{1}{\sqrt{n-\ell}} + p_0-\frac{1}{2}\right\rvert=   \eps+p_0+\frac{\euler}{\pi}\frac{1}{\sqrt{n-\ell}}.$$
\end{proof}
%$$\pr(E_4) = \pr((\lnot E_2\cap E_4) \cup (E_2 \cap E_4))  \leq  \pr(\lnot E_2)+\pr(E_4|E_2)\pr(E_2) \leq  p_0+\pr(E_4|E_2)\leq p_0+\eps.$$

The reason why the above protocol can terminate immediately after the $n^\text{th}$ block has been fetched is that even unintentional forks do not occur in $\nofork$, hence none of the $n$ blocks can be reversed.

%\subsection{General adversary}

\medskip

Let us now broaden Lemma \ref{claim:upbound1} to a model that captures Bitcoin. The model that we employ is the {\em Bitcoin backbone} formulation, that was proposed in \cite{EC:GKL15}. In this setting, an adversary with $p$ fraction of the mining power may be able to create more than $p$ fraction of the blocks in expectation, due to adversarial strategies that take advantage of blockchain forks (cf. \cite{FC:ES14,FC:SSZ16}). Yet, Assumption \ref{assume:maxprof} is reasonable in this model too, per the discussion in Section \ref{sec:profitmargin}.
\begin{overview}\label{overview:backbone}
\normalfont{
In the framework of \cite{EC:GKL15}, there are $N$ parties (or $N$ miners), and 
an adversary $\mc{A}$ is allowed to statically corrupt $t$ of them
and use their quota of hashing queries. %, which corresponds to a total of $t\cdot q$ queries.
Hash invocations in \cite{EC:GKL15} are modeled as random oracle queries, and hence the mining power
of $\mc{A}$ is $\beta\triangleq tqT/2^{\kappa}$, where $q$ is the quota of queries that
a miner has in each round of an execution of the backbone protocol, $T$ is the difficulty target, and $\kappa$ is
the length of the digest. The mining power of the honest miners is $\alpha\triangleq (N-t)qT/2^{\kappa}$,
and $\gamma\triangleq\alpha^2-\alpha$ represents a lower bound on the power of the honest miners after unintentional
forks are taken into consideration (see \cite[Section~4.1]{EC:GKL15} for further details).
}
\end{overview}
Thus, $\alpha$ and $\beta$ correspond to $1-p$ and $p$ in the model $\nofork$.

In \cite{EC:GKL15}, the adversary $\mc{A}$ together with the environment $\mc{Z}$ can orchestrate
an attack against the protocol %operating
%operation
execution
in any way feasible, independently
of the actual cost that the attack might be worth.
Therefore, we augment \cite{EC:GKL15} with the following definitions, that serve as the basis for the interpretations that we will soon provide.
%In reality however attackers are on a budget and their investment on mining equipment
%suggests that they expect revenue to be produced as a result of their mining operation. 
%Specifically, we use the following formalization: 

\begin{definition}
Let $\mc{B}_{B,\ell,L}$ be a predicate that is parameterized by a block $B$ and lengths $\ell$ and $L$. Given a chain $\chain$ and an adversary $\mc{A}$, let $\chain^{\mc{A}}$ denote an annotated chain that specifies for each of its blocks whether it was created by $\mc{A}$. We say that $\mc{B}_{B,\ell,L}$ is a bankruptcy predicate if for any adversary $\mc{A}$ and for any chain $\chain$, it holds that $\mc{B}_{B,\ell,L}(\chain^{\mc{A}})=1$ iff $\chain^{\mc{A}}$ contains a segment of $L$ consecutive blocks that starts with $B$, and $\mc{A}$ created at most $\ell$ out of these $L$ blocks.
\end{definition}

\begin{definition}\label{dfn:bankevent}
An execution of the backbone protocol is said to have a bankruptcy event w.r.t. an adversary $\mc{A}$ and a bankruptcy predicate $\mc{B}_{B,\ell,L}$ if it holds that at some round $r$ of the execution, every honest party $P$ adopts a chain $\chain^{\mc{A}}_P$ for which $\mc{B}_{B,\ell,L}(\chain^{\mc{A}}_P)=1$.
\end{definition}

Let us clarify that Definition \ref{dfn:bankevent} does not imply anything about the level of anonymity in the system. That is to say, an honest party does not necessarily know which blocks were created by $\mc{A}$, but due to the agreement among all honest parties we have that w.h.p. $\mc{A}$ will not collect more than $\ell$ rewards in a segment of $L$ blocks that starts at the block $B$.

\smallskip
The beacon protocol $\bprot$ is given in Figure~\ref{fig:bprot}. The following \claimtwo shows that an adversary who discards ``detrimental'' blocks will either fail to influence $\bprot$, or become bankrupt while trying. In the denotations of this \claimtwo, the event $E_{b}\cap\lnot E_{\ell}$ implies that the adversary continued to operate after she went bankrupt. While the Bitcoin backbone formulation itself does not prevent this possibility, one may consider an extended model in which \mbox{$\pr(E_{b}\cap\lnot E_{\ell})=0$}.
%The following Theorem shows that w.r.t. budget-restricted adversaries, a majority-based Beacon protocol can be secure against any adversarial strategy.

\begin{figure}[btb]
\begin{minipage}{\textwidth}
\begin{framed}
{
\small
\begin{center} Protocol $\bprot$ \end{center}
\begin{enumerate}
\item Fetch $n$ consecutive blocks $B_1,B_2,\ldots,B_n$ from the Bitcoin blockchain, such that $B_n$ has already been extended by $\kparam$ additional blocks.
\item Obtain a bit $b_i$ from each $B_i$ by taking the LSB of the symbol that is associated with $B_i$. 
\item Output the bit $\mathsf{majority}(b_1,b_2,\ldots,b_n)$.
\end{enumerate}
}
\end{framed}
\end{minipage}
\caption{The beacon protocol.}\label{fig:bprot}
\end{figure}

\begin{lemma}\label{claim:upbound2}
Let $\eps>0, n\geq 1, k\geq 1$ be arbitrary.
Following Overview~\ref{overview:backbone}, assume that the mining power of the adversary $\mc{A}_2$ is $\beta$, and that $\gamma = \lambda(1+\delta) \beta$ is satisfied for some $\delta\in(0,1), \lambda \in [1,\infty)$.
Denote $\ell\triangleq \floor{\frac{\eps}{2}\frac{\pi}{\euler}\sqrt{n-\sqrt{n}}}, L\triangleq \ceil{2(1-\frac{\delta}{3})^{-1} \lambda \ell}$.
Suppose that an honest party invokes the protocol $\bprot$.
Let $E_{b}$ be the bankruptcy event w.r.t. $\mc{A}_2$ and the predicate $\mc{B}_{B_1,\ell,L}$, and let $E_{\ell}\triangleq\{\mc{A}_2 \text{\ contributed at most $\ell$ of the blocks $B_1,B_2,\ldots,B_n$}\}$. If $\mc{A}_2$ never publishes any of the ``detrimental'' blocks that she creates, then $\bprot$ outputs a bit whose statistical distance from a uniform random bit is $\eps+ \pr(E_{b}\cap\lnot E_{\ell})+(\Omega(\sqrt{n}))^{-1} + e^{-\Omega(\delta^2 \ell)} + e^{-\Omega(\kparam)}$ at the most.
\end{lemma}
\begin{proof}
Since $\mc{A}_2$ discards all of the ``detrimental'' blocks that she obtains from the random oracle, she can be regarded as an adversary whose mining power is $\beta'=\nicefrac{\beta}{2}$. Thus, for $\lambda'=2\lambda$, the equality $\gamma = \lambda'(1+\delta) \beta'$ holds.

Let $E_{s}=\{\mc{A}_2$ contributed at least as many blocks among $B_1,B_2,\ldots,B_n$ as the difference between ``helpful'' and ``detrimental'' blocks that the honest miners contributed among $B_1,B_2,\ldots,B_n\}$. We have ``s'' in $E_{s}$ to indicate that $\mc{A}_2$ succeeded to influence the output bit.

Since we defined $\ell$ with $\frac{\eps}{2}$ instead of $\eps$, Lemma \ref{lemma:majextr} implies that $\pr(E_{s}|E_{\ell})\leq\frac{\eps}{2}$ holds. 

For the chain of consecutive blocks $B_1,B_2,\ldots,B_n,\ldots$ that an honest party fetches and thereby derives the output of the protocol, let $B_m$ be the $\ell^\text{th}$ block that $\mc{A}_2$ contributed. In other words, $\mc{A}_2$ contributed $\ell$ of the first $m$ blocks, and the honest miners created the other $m-\ell$ blocks. Denote $E_{q}\triangleq\{m\geq L\}$. According to \cite[Theorem 10]{EC:GKL15}, the chain quality property entails that $\pr(E_{q})=1-e^{-\Omega(\delta^2 m)} \geq 1-e^{-\Omega(\delta^2 \ell)}$. In fact, \cite[Theorem~10]{EC:GKL15} guarantees an even stronger property, see Remark \ref{rem:supply}. The letter ``q'' in $E_{q}$ refers to the quality of the segment $B_1,B_2,\ldots,B_L$. 

Let $r$ denote the round at which the protocol $\bprot$ terminated. Let us define $E_{c}\triangleq\{\text{At round} \ r, \ \text{the block} \ B_n \ \text{is included in every chain of of every honest party}\}$. Due to the common prefix property \cite[Theorem~9]{EC:GKL15}, it holds that $\pr(E_{c})=1-e^{-\Omega(\kparam)}$. We have ``c'' in $E_{c}$ to indicate that the common prefix $B_1,B_2,\ldots,B_n$ is in consensus.

Let us consider two cases.

For the first case, we assume that $L\leq n$. Here, it holds that $\lnot E_{b} \cap E_{q} \subseteq \lnot E_{c}$, because $E_{q}\cap E_{c} \subseteq E_{b}$. The reason that $E_{q}\cap E_{c} \subseteq E_{b}$ holds is as follows. If $E_{c}$ occurred, then for every honest party $P$ we have that the chain $\chain_{P,r}$ that $P$ adopted at round $r$ includes $B_1,B_2,\ldots,B_n$. If $E_{q}$ occurred, then for any chain $\chain$ that contains $B_1,B_2,\ldots,B_n$ we have $\mc{B}_{B_1,\ell,L}(\chain^{\mc{A}_2})=1$. Hence, if both $E_{q}$ and $E_{c}$ occurred, then at round $r$ every honest party $P$ adopted a chain $\chain_{P,r}$ such that $\mc{B}_{B_1,\ell,L}(\chain^{\mc{A}_2}_{P,r})=1$. This implies that $E_{b}$ occurred.

By the law of total of total probability, we obtain
\begin{eqnarray*}
\pr(\lnot E_{b}\cap E_{s}) &=& \pr( (\lnot E_{b} \cap E_{s} \cap E_{\ell}) \cup (\lnot E_{b}\cap E_{s} \cap\lnot E_{q}) \cup (\lnot E_{b}\cap E_{s} \cap \lnot E_{\ell} \cap E_{q})) \\
& \leq & \pr(E_{s}\cap E_{\ell}) + \pr(\lnot E_{q})  + \pr(\lnot E_{b}\cap E_{q}) \\
& \leq & \pr(E_{s}|E_{\ell})\pr(E_{\ell}) + e^{-\Omega(\delta^2 \ell)} + \pr(\lnot E_{c}) \\
& \leq & \pr(E_{s}|E_{\ell}) + e^{-\Omega(\delta^2 \ell)} + e^{-\Omega(\kparam)}\leq \hspace{2pt}\frac{\eps}{2}+ e^{-\Omega(\delta^2 \ell)} + e^{-\Omega(\kparam)}.
\end{eqnarray*}

For the second case, we assume that $L> n$. Here, $E_{q}\subseteq E_{\ell}$. Therefore,
\begin{eqnarray*}
\pr(E_{s}) &=& \pr( (E_{s}\cap E_{\ell}) \cup \pr(E_{s}\cap\lnot E_{\ell}) \\
& \leq & \pr(E_{s}|E_{\ell})\pr(E_{\ell}) + \pr(E_{s}\cap\lnot E_{\ell}) \\
& \leq & \pr(E_{s}|E_{\ell}) + \pr(\lnot E_{\ell}) \\
& \leq & \frac{\eps}{2} + \pr(\lnot E_{q}) \leq\hspace{2pt} \frac{\eps}{2} + e^{-\Omega(\delta^2 \ell)}.
\end{eqnarray*}

Thus, the second case implies an even smaller bound on $\pr(\lnot E_{b}\cap E_{s})$.

Let $E_0\triangleq\{\mathsf{majority}(b_1,b_2,\ldots,b_n)=0\}$. Since $\lnot E_{s}$ implies that more than half of the bits $b_1,b_2,\ldots,b_n$ are uniform random, and since a tie among $\nicefrac{n}{2}$ uniform random bits occurs with probability $\frac{\euler}{\pi}\frac{1}{\sqrt{n/2}}$ at the most (cf. Lemma \ref{lemma:majextr}), we have $\pr(\lnot E_0 \cap \lnot E_{s})\leq\frac{1}{2}+\frac{\euler}{\pi}\frac{1}{\sqrt{n/2}}$. Therefore,
\begin{eqnarray*}
\pr(\lnot E_0 \cap (\lnot E_{b}\cup E_{\ell})) &=&  \pr(\lnot E_0 \cap (\lnot E_{b}\cup E_{\ell})\cap E_{s})+\pr(\lnot E_0 \cap (\lnot E_{b}\cup E_{\ell})\cap\lnot E_{s})\\
&\leq&  \pr((\lnot E_{b}\cup E_{\ell})\cap E_{s})+\pr(\lnot E_0 \cap \lnot E_{s})\\
&\leq& \pr(E_{s}|E_{\ell})\pr(E_{\ell})+\pr(\lnot E_{b}\cap E_{s})+\pr(\lnot E_0 \cap \lnot E_{s})\\
& \leq & \frac{\eps}{2}+ \bigl(\frac{\eps}{2}+ e^{-\Omega(\delta^2 \ell)} + e^{-\Omega(\kparam)}\bigr) + \bigl(\frac{1}{2}+\frac{\euler}{\pi}\frac{1}{\sqrt{n/2}}\bigr).\\
\end{eqnarray*}
This implies
\begin{eqnarray*}
\pr(E_0)+\pr(E_{b}\cap\lnot E_{\ell})\geq\pr(E_0 \cup (E_{b}\cap\lnot E_{\ell})) \geq \frac{1}{2}-\frac{\euler}{\pi}\frac{1}{\sqrt{n/2}} -\eps- e^{-\Omega(\delta^2 \ell)} - e^{-\Omega(\kparam)}.\\
\end{eqnarray*}
Since $\mc{A}_2$ tries to bias the output towards $1$, we obtain that the statistical distance between the output bit and a uniform random bit is at most
\begin{eqnarray*}
\Bigl\lvert\pr(E_0)-\frac{1}{2}\Bigr\rvert
&\leq& \Bigl\lvert\frac{1}{2}-(\eps+\pr(E_{b}\cap\lnot E_{\ell})+\frac{\euler}{\pi}\frac{1}{\sqrt{n/2}}+ e^{-\Omega(\delta^2 \ell)}+ e^{-\Omega(\kparam)})-\frac{1}{2}\Bigr\rvert \\
&=& \eps+\pr(E_{b}\cap\lnot E_{\ell})+\frac{\euler\sqrt{2}}{\pi}\frac{1}{\sqrt{n}}+ e^{-\Omega(\delta^2 \ell)}+ e^{-\Omega(\kparam)}.\\
\end{eqnarray*}

\end{proof}

\begin{remark}
\claimtwo \ref{claim:upbound2} also guarantees agreement on the output bit among all the honest miners. Let $j$ denote the distance from the genesis block to $B_1$, and consider the event $E_a\triangleq\{B_1,B_2,\ldots,B_n \text{\ reside in every chain of length $j+n+k$ or more that an honest miner adopts}\}$. According to \cite[Theorem~9]{EC:GKL15}, we have $\pr(E_a)=1-e^{-\Omega(\kparam)}$. Therefore, $\pr(\lnot E_0\lor\lnot E_a)\leq \pr(\lnot E_0)+\pr(\lnot E_a)\leq \frac{1}{2}+\eps+\pr(E_{b}\cap\lnot E_{\ell})+(\Omega(\sqrt{n}))^{-1}+ e^{-\Omega(\delta^2 \ell)} + \bcancel{2} e^{-\Omega(\kparam)}$. This implies that the output bit of the beacon is publicly verifiable in the following sense: any honest party can verify the validity of some $n+k$ consecutive blocks and compute the bit that the first $n$ blocks derive, and have the assurance that this bit is close to uniform and that w.h.p. every other honest party will derive the same bit.
\end{remark}

It should be noted that the hidden constants in the terms $e^{-\Omega(\kparam)}$ and $e^{-\Omega(\delta^2 \ell)}$ depend on a parameter $f$ that measures the synchronicity of the honest miners, as can be seen by inspecting \cite[Theorem~9]{EC:GKL15} and \cite[Theorem~10]{EC:GKL15}. However, with $10$ minutes average block interval as in Bitcoin, and reasonable assumptions on the network propagation latency \cite{P2P:DW13}, the probability of forks by the honest miners is less than $3\%$, and therefore the constant $f$ is quite modest. % (see also \cite{AggelosPreprint}). %TODO: improve

\begin{observation}\label{rem:nonstop}
We may suppose w.l.o.g. that $\mc{A}_2$ ceaselessly tries to create blocks for as long as she is not bankrupt, since an adversary $\mc{A}_0$ who opts to remain idle (instead of submitting queries to the random oracle at some rounds of the backbone protocol execution) gains nothing by doing so. This is because the majority function is influenced only by the total amount of blocks that $\mc{A}_0$ contributes among the $n$ blocks, and therefore it is better to create the blocks as soon as possible. Hence, $\mc{A}_2$ is at least as powerful w.r.t. influencing the output bit of $\bprot$ as $\mc{A}_0$ is.
\end{observation}

The above observation leads us to the following interpretation.

\begin{interpretation}\label{interpret:a1}
\normalfont{
Consider the adversary $\mc{A}_2$ of \claimtwo \ref{claim:upbound2}. Suppose that every block $B$ that $\mc{A}_2$ creates such that $B$ will be included in every chain of every honest party from a certain round, $\mc{A}_2$ earns $x$ coins. 
Suppose that for every block that will be included in every chain of every honest miner from a certain round, $\mc{A}_2$ loses either $y(\beta)$ or more coins.
Because Observation \ref{rem:nonstop} allows us to assume that $\mc{A}_2$ attempts to create every block, $y(\beta)$ represents the operating expenses of $\mc{A}_2$ for each block that she either succeeded or failed to create, and $x$ represents the reward that $\mc{A}_2$ earns for a block that she successfully created. For $\beta=\beta' N qT/2^{\kappa}$, denote $z_{\beta}\triangleq \beta' x - y(\beta)$. The quantity $z_{\beta}$ can be regarded as an upper bound on the ``honest'' profit margin of $\mc{A}_2$ per block, i.e., it is an approximation of the profit margin that $\mc{A}_2$ can initially expect in case she extends only the head of the chain. It is therefore reasonable to assume that $z_\beta>0$. Denote $w_{\beta}\triangleq \frac{1}{2} \frac{1}{\lambda}(1-\frac{\delta}{3}) x - y(\beta)$. As implied by \cite[Theorem~10]{EC:GKL15}, the quantity $w_{\beta}$ can be viewed as an upper bound on profit margin of $\mc{A}$ while she deploys any adversarial strategy but discards half of the blocks that she creates. Hence, in case $\mc{A}_2$ has an initial reserve of $R$ coins but she was only able to create at most $\ell$ of the first $L$ blocks in $\bprot$, she would be left with a supply of $R+\ell\cdot x - L\cdot y(\beta)=R+\ell(1-\frac{\delta}{3})\frac{1}{\lambda'} w_\beta$ or less coins. Suppose that $z_\beta$ is small enough so that $w_{\beta}<0$. Suppose that $\ell$ is large enough so that $R+\ell(1-\frac{\delta}{3})\frac{1}{\lambda'} w_\beta<0$ holds. We can now see that the formal bankruptcy event coincides with the interpretation that $\mc{A}_2$ depleted her entire initial budget of $R$ coins. Specifically, using the denotations of \claimtwo~\ref{claim:upbound2}, the occurrence of $E_{q}\cap E_{c}$ is interpreted to mean that $\mc{A}_2$ ran out of coins while trying to contribute as many blocks as possible among the first $L$ blocks. This is because $E_{c}$ implies that all honest miners would agree that the only coins that $\mc{A}_2$ earned are these $\leq \ell x$ coin rewards, and thus $\mc{A}_2$ would have no other coins of value.}
\end{interpretation}
For instance, consider $\delta=\frac{1}{100},\beta=\nicefrac{N}{5}\cdot qT/2^{\kappa},\lambda=3.9$, and let us suppose that $\mc{A}_2$ can deploy a strategy that enables her to earn up to $\frac{1}{3.9}(1-\frac{1}{300}) x - y(\beta) > z_{\beta}=\frac{1}{5}x-y(\beta)$ coins per block. However, when $\mc{A}_2$ discards half of her solved blocks, we expect her profits per block to be $\frac{1}{7.8}\frac{299}{300} x - y(\beta)$ at the most, for any adversarial strategy that she deploys.

\medskip
At this point, let us consider the broader range of adversaries. The next remark will come useful in the analysis that follows it.

\begin{remark}\label{rem:supply}
A closer examination of \cite[Theorem~10]{EC:GKL15} reveals that the mining power of $\mc{A}_2$ will allow her to create a supply of $\ell'$ blocks such that she will be able to include either all of these $\ell'$ blocks or less in a chain that consists of $L=(1-\frac{\delta}{3})^{-1} \lambda \ell$ consecutive blocks in total, and the probability that the size of this supply satisfies $\ell'>\ell$ is $e^{-\Omega(\delta^2 L )}$, and hence also $e^{-\Omega(\delta^2 \ell )}$ at the most\footnote{The parameter $\ell$ of \cite[Theorem~10]{EC:GKL15} is denoted by $L$ in our case.}.
Therefore, we can regard $\mc{A}_2$ as an even more powerful adversary who operates in two phases. In the first phase, $\mc{A}_2$ creates a private chain $\chain$ that starts from the genesis block, by submitting queries to the random oracle. Then, at whichever round that $\mc{A}_2$ desires, she switches to the second phase, where she picks any chain $\chain'$ of an honest party, and is then granted the unusual power of being able to instantly find a nonce that modifies the second block of $\chain$ so that it would extend the last block of $\chain'$ rather than the genesis block. Then, $\mc{A}_2$ broadcasts the valid chain $\chain'\rightarrow\chain$, and from then on she may only extend chains that include $\chain$.
\end{remark}

We are now ready to show that \claimtwo \ref{claim:upbound2} together with Assumption \ref{assume:maxprof} guarantee that when $n$ is large enough, {\em any} budget-restricted adversary will fail to influence the beacon. Keep in mind that this requires a realistic profits bound, otherwise the lower bound of Section \ref{sec:lowerbound} would apply (as illustrated in Example \ref{example:profits}).

\begin{interpretation}\label{interpret:a2}
\normalfont{
Consider an adversary $\mc{A}_3$ who may deploy any possible strategy, which implies that she may publish ``detrimental'' blocks. Suppose that the mining power of $\mc{A}_3$ is also $\beta$, and that she purchased this mining power with $R_1$ coins. Suppose that $\mc{A}_3$ has an additional reserve of $R_2$ coins. For the parameter $n$ that $\bprot$ is invoked with, suppose that the adversary $\mc{A}_2$ has a reserve of $R_2+\maxcoins(R_1,n)$ coins. Following Interpretation \ref{interpret:a1}, let us note that we expect $z_\beta$ to diminish over time, as can be inferred by Assumption \ref{assume:maxprof}. One way to think of it is that $y(\beta)$ grows because $\mc{A}_2$ needs to keep expending more resources in order to maintain $\beta$ level of mining power relative to the honest miners, by buying new mining equipment and repairing old equipment. Let us assume that $n_0$ and $\ell_0\triangleq \floor{\frac{\eps}{2}\frac{\pi}{\euler}\sqrt{n_0-\sqrt{n_0}}}$ are large enough so that the condition
$R_2+\maxcoins(R_1,n_0)+\ell_0(1-\frac{\delta}{3})\frac{1}{\lambda'} w_\beta<0$ holds. 
%We may regard $\mc{A}_3$ as an adversary whose only objective is the influence the output of $\bprot$, as an adversarial strategy that takes other objectives into account can only make $\mc{A}_3$ weaker in this context.
The adversary $\mc{A}_3$ may opt to include ``detrimental'' blocks that she created and increase her coins reserve this way, thereby allowing her to carry on playing before she runs out of coins, albeit at the cost of influencing the beacon in the wrong direction. However, according to Assumption \ref{assume:maxprof}, the maximal amount of coins that $\mc{A}_3$ can earn is $\maxcoins(R_1,n_0)$. In view of Remark~\ref{rem:supply}, we can regard $\mc{A}_2$ as building one secretive private chain. This means that it is never beneficial for $\mc{A}_2$ to broadcast a shorter chain that includes ``detrimental'' blocks that she solved in order to win an intermediate race against the honest miners. Therefore, the {\em only} possible advantage that $\mc{A}_2$ would have had if she was capable of including ``detrimental'' blocks in her chain is the ability to earn extra revenues in order to avoid going under budget. Since the initial budget of $\mc{A}_2$ has $\maxcoins(R_1,n_0)$ more coins than the initial budget of $\mc{A}_3$, it follows that although $\mc{A}_2$ must discard all of her ``detrimental'' blocks, it is less probable that she will run out of coins before $\mc{A}_3$ runs out of coins, than vice verse. Thus, it is more likely than not that $\mc{A}_2$ will contribute at least the same amount of ``helpful'' blocks as $\mc{A}_3$ contributes. Since $\mc{A}_2$ achieves a bias of $\eps+\pr(E_{b}\cap\lnot E_{\ell})+(\Omega(\sqrt{n}))^{-1}+ e^{-\Omega(\delta^2 \ell)} + e^{-\Omega(\kparam)}$ at the most, it thus follows that the bias that $\mc{A}_3$ achieves is also $\eps+\pr(E_{b}\cap\lnot E_{\ell})+(\Omega(\sqrt{n}))^{-1}+ e^{-\Omega(\delta^2 \ell)} + e^{-\Omega(\kparam)}$ or less.
Let us draw attention to the facts that $\mc{A}_2$ goes bankrupt except with negligible probability, while $\mc{A}_3$ does not necessarily go bankrupt. Yet, we have that both $\mc{A}_2$ and $\mc{A}_3$ are destined to fail in their attempt to influence the output bit of $\bprot$, unless they can continue to operate after they become bankrupt.
}
\end{interpretation}
Hence, Interpretation \ref{interpret:a2} demonstrates that for any adversary \adv with an arbitrary initial budget, we can set $n$ and $k$ to be large enough so that the bias of the beacon is arbitrarily small.
Concretely, to achieve bias of at most $\eps$, \claimtwo~\ref{claim:upbound2} requires $e^{-\Omega(\delta^2 \ell)}\leq\eps$ for $\ell= \floor{\frac{\eps}{2}\frac{\pi}{\euler}\sqrt{n-\sqrt{n}}}$. Since $e^{-\eps\sqrt{n}}\leq\eps \Leftrightarrow n\geq\frac{1}{\eps^2}\log^2\frac{1}{\eps}$, Interpretation~\ref{interpret:a2} implies that $n=O(\frac{1}{\eps^2}\log^2\frac{1}{\eps})$ blocks are enough, unless the budget and profits of \adv are large enough so that $R_2+\maxcoins(R_1,n)$ is greater than $\approx \frac{1}{2\lambda}(-w_\beta)\ell\approx \frac{1}{2\lambda}(-w_\beta)\log\frac{1}{\eps}$. Notice that in case $\maxcoins(R_1,n)=\Omega(\sqrt{n})$, Interpretation \ref{interpret:a2} does not imply that an arbitrarily small $\eps$ bias can be achieved. However, per Section \ref{sec:profitmargin}, it is reasonable to assume that $\maxcoins(R_1,n)\leq 2 R_1$.
%the maximal profits of $\mc{A}_2$ during an $n$ blocks chain growth are

\medskip
Let us note that Observation \ref{rem:nonstop} should hold even in more complex models. For instance, it would hold in a model in which an adversary \adv can gain a windfall of coins due to an external event that occurs with some likelihood. This is because of Remark \ref{rem:supply} and the fact that the majority extractor is influenced only by the total amount of locations that \adv controls, which together imply that \adv can play until she runs out of coins and then wait for the external event to occur and thereby negate her bankruptcy. Furthermore, $z_{\beta}>0$ also implies that Observation~\ref{rem:nonstop} would hold in complex models.

\medskip
A superficial review of \claimtwo \ref{claim:upbound2} may lead one to believe that the assumption that the adversary has less than $\frac{1}{2}$ of the mining power is inessential. While the constraint $\gamma = \lambda(1+\delta) \beta$ is required in order to invoke \cite[Theorem 9]{EC:GKL15} and \cite[Theorem 10]{EC:GKL15}, the transformation from $\lambda$ to $\lambda'=2\lambda$ may seem to imply that an adversary with mining power $\beta\approx\nicefrac{2}{3}\cdot{N}qT/2^{\kappa}$ can also be coped with. However, let us show that Assumption~\ref{assume:maxprof} may no longer make sense under these conditions. Let $\frac{1}{2}<\beta'<\frac{2}{3}$. An adversary $\mc{A}_4$ with mining power $\beta=\beta' N qT/2^{\kappa}$ succeeds with probability $1$ to create a private longest chain, by winning the race against the honest miners \cite{MAN:r12,MAN:n08}. This longest chain can be arbitrarily large, and it consist of blocks that only $\mc{A}_4$ contributed. Note that $\mc{A}_4$ would not discard any of the blocks that she creates while working on this private chain. Hence, following Interpretation~\ref{interpret:a1}, $\mc{A}_4$ collects a reward of $x$ coins per each of the blocks in this private chain. Although the honest profit margin is $z_\beta=\beta' x-y(\beta)$, these observations imply that the effective profit margin of $\mc{A}_4$ reaches $x-y(\beta)$ when she mounts an attack. Suppose that $\mc{A}_4$ acquired her $\beta$ mining power with $t$ coins. If Assumption~\ref{assume:maxprof} entails that $\maxcoins(t,n)=O(t)$ even in the case that $\mc{A}_4$ collects {\em all} of the block rewards, then Interpretation~\ref{interpret:a2} still holds. On the other hand, if we make the plausible assumption that every block costs $\mc{A}_4$ at most $x-c$ coins to produce, where $c$ is a constant, then Example \ref{example:profits} demonstrates that $\mc{A}_4$ can defeat $\bprot$. Overall, it may indeed be true that an adversary with up to $\approx\nicefrac{2}{3}$ of the mining power is still weak enough in a practical setting~(cf. Section~\ref{sec:profitmargin}), but such an adversary is likely to be too powerful when $\bprot$ is invoked with a sufficiently large $n$.

As a side note, let us remark that a similar reasoning to the above shows that an adversary $\mc{A}_5$ who operates in the Bitcoin backbone model with an infinite budget and more than $\nicefrac{1}{2}$ of the mining power can bias the output bit by the largest possible amount (i.e., a larger bias than that of the adversary in Section~\ref{sec:lowerbound}). The adversary $\mc{A}_5$ can create her own private chain that has at least $n+\kparam$ blocks and is longer than the chains that the honest parties have, and publish it only in case it derives the desired output bit for the beacon. This event happens with probability $\nicefrac{1}{2}$, because $\mc{A}_5$ succeeds to create the longest chain with probability $1$. Otherwise, $\mc{A}_5$ discards her chain, and proceeds to create a fresh longest chain that starts from where the discarded chain started. Even though it should take $\mc{A}_5$ much longer to create the second chain, she still succeeds with probability $1$, and therefore succeeds to obtain the desired beacon output with probability $\nicefrac{1}{2}$. Otherwise, she discards her chain again, and repeats the process. Hence, overall, $\mc{A}_5$ succeeds to set the beacon to the output that she desires with probability $\frac{1}{2}+\frac{1}{4}+\frac{1}{8}+\ldots=1$.

\medskip
While \claimtwo \ref{claim:upbound2} and Interpretation \ref{interpret:a2} establish that $\bprot$ is secure against any adversarial strategy, it can be useful to consider practical strategies that an adversary may use. In a usual setting, selfish mining strategies \cite{FC:ES14,EPRINT:L13,FC:SSZ16} may become profitable only after at least one difficulty readjustment occurs (in Bitcoin this happens after a $2016$ blocks window that takes $\approx$ 2 weeks). This is because a selfish miner would have had less orphaned blocks if she dedicated her entire mining power to generating blocks at the head of the honest chain.
However, when the goal of the adversary is to influence the beacon rather than to earn revenues, a selfish mining strategy can be effective even before a difficulty readjustment occurs, because the proportional amount of blocks that the adversary contributes relative to the honest miners will already be larger.
That being said, our utilization of the extraordinary adversary that is described in Remark~\ref{rem:supply} was only for the purpose of the analysis, and in reality an adversary with less than $\nicefrac{1}{2}$ of the mining power cannot sustain a private chain for long. Thus, when considering an adversary with $\beta$ mining power and $\gamma = \lambda(1+\delta) \beta, \lambda\geq 1$ as in \claimtwo~\ref{claim:upbound2}, it is unlikely that this adversary can create $\nicefrac{1}{\lambda}$ portion of the blocks. To give a concrete example, the basic selfish mining strategy of \cite{FC:ES14} implies that an adversary with $\nicefrac{1}{3}$ of the mining power and no connectivity advantage can create $38.4\%$ of the blocks.

\subsection{Discussion of mining profitability}\label{sec:profitmargin}
Since the market entry costs for Bitcoin mining are low, the expected profit margin cannot be high. The reason for this is that if anyone could purchase a Bitcoin mining machine that creates money ad infinitum, then everyone would become a Bitcoin miner. Actually, it is also unreasonable to assume that Bitcoin mining is necessarily risk-free, as in reality it is in fact unprofitable to be a Bitcoin miner sometimes (see \cite{unprofit0,unprofit1,unprofit2,unprofit4,unprofit5} for some analysis and examples). 

Hence, an unexcessive growth rate for $\maxcoins(t,n)$ that satisfies $\forall n: \maxcoins(t,n)\leq 2t$ is more than reasonable, as most miners would be quite happy to double the return on their investment. Of course, our analysis can also accommodate a more lucrative bound, for example $\forall n: \maxcoins(t,n)\leq 100t$. In fact, Interpretation~\ref{interpret:a2} shows that our analysis holds for $\maxcoins(t,n)=o(\sqrt{n})$, though it is more natural to assume that the profits bound is a function of $t$.

%wanted: also, practical p=1/10 n=year consider n bounded not infinite

The economic reasoning can also be elaborated upon as follows. Suppose that Alice buy a citrus juicer and starts a lemonade stand. The price at which Alice buys lemons is lower than the price at which she sells the lemonade, hence Alice makes a profit with each sale. While Alice's initial investment in the citrus juicer enables her to continuously accumulate profits, we do not expect Alice to become a billionaire as a result of this process. Our pessimistic prediction stems from the observations that the citrus juicer erodes over time, that Alice's profit margin is unlikely to be excessive, and that Alice's wealth will be measured relative to her competitors rather than in absolute terms. If Alice decides to drop the lemonade business and become a Bitcoin miner instead, the same reasoning still applies.

Hence, even if a Bitcoin miner re-invests some of her profits by buying additional mining equipment, we can expect this process to converge towards some bound on the overall amount of profits. While \claimtwo \ref{claim:upbound2} assumes that the adversary's mining power $\beta$ is constant, we can instead interpret $\beta$ as the maximal mining power that the adversary can obtain as a result of the re-investments process.

To take some exemplary figures, a miner may have an average cost of at least $98$ coins in terms of power consumption mining equipment erosion until she succeeds to solve a block, and the reward that she earns from each solved block is $100$ coins on average. If we examine Lemma~\ref{claim:upbound1}, we see that for increasingly greater values of $n$ it is indeed the case that $d$ can be larger while still satisfying the inequality  $\delta\frac{1}{p'}\ell \leq n-d$. However, the gains that the adversary achieves with a large $d$ can be disregarded due to Assumption~\ref{assume:maxprof}. This can be thought of as quite reasonable since the profit margin is modest ($100-98=2$ or less coins profit per solved block), as opposed to mining at a loss by discarding half of the solved blocks, which will lead to a quick deterioration of funds (a loss of $50-98=-48$ coins per solved block).
%Thus, in practical setting, %Assumption~\ref{assume:maxprof}
%This can be interpreted as saying that 
This can be taken to mean that 
%$\maxcoins$ can also be regarded as a bound on the amount of revenues over a reasonably long time period $n$. That is to say, if the profit margin is low enough, then one can assume that
the adversary does not earn a substantial enough amount of coins by including many ``detrimental'' blocks before the $n$ blocks that the beacon utilizes are created.
Thus, in a practical setting, it is likely that the adversary $\mc{A}_2$ of Interpretation~\ref{interpret:a1} is already as powerful as the general adversary $\mc{A}_3$ of Interpretation~\ref{interpret:a2}.

\smallskip
Let us provide a counterexample that shows that an unreasonable profits structure implies that the adversary is powerful enough to defeat any beacon protocol.

\begin{example}\label{example:profits}
\normalfont Consider an unrealistic setting where the adversary's profit margin per block is constant during any arbitrarily long period of time. Hence, in particular, a condition of the form $\exists c\ \forall n: \maxcoins(t,n)\leq c\cdot t$ does not hold. In this setting, it is impossible to have a beacon that outputs a bit whose statistical distance from uniform is arbitrarily close to $0$, even in the forkless model $\nofork$, and even if the adversary has a finite budget that prevents her from operating at a loss indefinitely. To see this, suppose for example that the reward per block is $x=60$ coins, the adversary has $p=\nicefrac{1}{10}$ of the mining power, and attempting to produce each block costs the adversary $y_p=5$ coins. Thus, the profit margin is $z_p=60\cdot\nicefrac{1}{10}-5=1$ coin per block. On the other hand, operating at a loss by discarding all ``detrimental'' blocks implies that the adversary has a cost of $w_p=60\cdot\nicefrac{1}{19}-5>60\cdot\nicefrac{1}{20}-5=-2$ coins per block on average. Therefore, one straightforward strategy for the adversary in this case is to operate honestly during the first $\nicefrac{n}{2}$ blocks, and discard all of her ``detrimental'' blocks in the next $\nicefrac{n}{4}$ blocks. Since $w > -2z$, w.h.p. the adversary will earn enough coins during the first $\nicefrac{n}{2}$ blocks so that she would not run out of funds while operating at a loss in the following $\nicefrac{n}{4}$ blocks. The beacon protocol should thus be secure against a non-oblivious symbol-fixing source that controls $\approx \nicefrac{n}{40}$ of the locations. However, our lower bound in Section \ref{sec:lowerbound} shows that an adversary who controls a constant fraction $\widehat{p}=\nicefrac{1}{40}$ of the locations can bias the output bit of the beacon by at least an $\nicefrac{1}{6}\cdot \widehat{p}=\nicefrac{1}{240}$ amount.
\end{example}

In fact, the above example implies that any constant production cost $y_p$ and positive profit margin $z_p$ make it impossible to achieve an arbitrarily small bias. In the real world it may be possible to have an investment that generates constant profits in nominal terms, though one should factor inflation-adjusted prices in that case.

%\newpage
\def\hybprot{\pi_{\text hyb}}
\def\hybadv{\mc{A}_{\text hyb}}
\def\com{\mathsf{com}}
\def\fcompen{\mathcal{F}_{\mathrm{com}}^\star}

\section{Hybrid protocol}\label{sec:hybrid}
%We now consider the combination of reliance of designated parties and randomness extraction from the Bitcoin public ledger.
We now consider a protocol that combines reliance on designated parties and randomness extraction from the Bitcoin public ledger. In this setting, the designated parties are hopefully reputable and honest, and thus the goal is to have a protocol that makes use of the designated parties in order to output a bit that has even less bias when compared to the $\bprot$ protocol.

\smallskip
An ideal objective for a hybrid protocol is of the following form:
\begin{itemize}
\item If the number of designated parties who are honest is above some threshold, then the output bit has less bias than that of $\bprot$.
\item Otherwise, the bias of the output bit is not worse than that of $\bprot$.
\end{itemize}
Unfortunately, this objective is bound to fail, barring additional assumptions. To see that, consider for example an adversary who corrupts {\em all} of the designated parties. If this adversary has $\beta$ mining power, then she should get some advantage relative to an adversary with $\beta$ mining power who tries to defeat $\bprot$, because the hybrid protocol takes into account the actions of the designated parties. Thus, we present a protocol that does not necessarily accomplish the ideal goal. 

In a similar fashion to other blockchain-based protocols that impose fairness \cite{SP:ADMM14,C:BenKum14,CCS:KumBen14,EC:KZZ15,SP:KMSWP15}, our hybrid protocol makes use of a {\em commitment with penalty} functionality $\fcompen$. The $\fcompen$ functionality allows a party $P_i$ to lock $q$ coins of hers, so that she can take back possession of these $q$ coins if and only if she reveals a certain decommitment before a time limit $\tau$. In case the limit $\tau$ passed and $P_i$ did not reveal the decommitment, the $q$ coins are destroyed. Due to the \texttt{OP\_CHECKLOCKTIMEVERIFY} \cite{MISC:CLTV} softfork, the $\fcompen$ functionality can be realized in Bitcoin in a straightforward manner. See Code~\ref{code:hiCLTV} for a high-level pseudocode of $\fcompen$, and Code~\ref{code:lowCLTV} for the actual Bitcoin script that corresponds to it.

\begin{algorithm}
\floatname{algorithm}{Code}
\caption{~~~Pseudocode of CLTV-based $\fcompen$}
\smallskip
\begin{algorithmic}[1]
\If {block\# $>$ $\tau$}
\State\Return {False}\Comment{the $q$ coins are unspendable}
\Else
\State $P_i$ can spend the $q$ coins by
\State \qquad signing with $sk_i$\Comment{the public key $pk_i$ of $P_i$ is hardcoded}
\State \qquad \textbf{and}
\State \qquad supplying $w$ such that $\hash(w)=c_i$\Comment{the commitment $c_i$ is hardcoded}
\EndIf
\end{algorithmic}
\label{code:hiCLTV}
\end{algorithm}

\begin{algorithm}
\floatname{algorithm}{Code}
\caption{~~~Bitcoin script of CLTV-based $\fcompen$}
\smallskip
\texttt{\ <$\tau$> CHECKLOCKTIMEVERIFY IF HASH256 <$c_i$> EQUALVERIFY <$pk_i$> CHECKSIGVERIFY ENDIF}
\label{code:lowCLTV}
\end{algorithm}

The hybrid protocol is presented in Figure \ref{fig:hyb}. The index $u_1$ specifies an agreed upon starting block $B_{u_1}$. As we discuss below, a sensible choice for the function $f$ can be $f(x_1,x_2,\ldots,x_m)=\mathsf{majority}(x_1,x_2,\ldots,x_m)$ or $f(x_1,x_2,\ldots,x_m)=x_1\oplus x_2\oplus\cdots\oplus x_m$. The parameter $t$ specifies a segment length in which each designated party $P_i\in\{P_1,P_2,\ldots,P_m\}$ is required to react, as otherwise $P_i$ is regarded as a corrupt party who aborted. The purpose of the parameter $k$ is to handle the possible re-organizations of the chain history by ensuring w.h.p. an agreement on a common prefix among the honest miners, in the beacon phase as well as in the commitments/decommitments phases. 

\begin{figure}[btb]
\begin{minipage}{\textwidth}
\begin{framed}
{
\small
\begin{center} Protocol $\hybprot$ \end{center}
\begin{itemize}
\item For each round $j\in[r]$:
\begin{itemize}
\item Denote $u'_j\triangleq u_j+t+k,\ u''_j\triangleq u'_j+n+k$.
\item \label{hybitem:initfair} {\em Commitments phase.} Before the block $B_{u_j}$ is extended by $t$ extra blocks:
\begin{itemize}
\item For each $i\in[m]$: $P_i$ picks a uniform random bit $d_i$, and invokes $\fcompen$ to broadcast a transaction that locks $q$ coins of hers, specifies a commitment $c_i=\com(d_i)$, and specifies $u''_j$ as the limit.
\end{itemize}
\item \label{hybitem:beacon} {\em Beacon phase.}
\begin{itemize}
\item Invoke $\bprot$ by waiting for $B_{u'_j}$ to be extended by $n+k$ new blocks $B_{u'_j+1},B_{u'_j+2},\ldots,B_{u''_j}$, thus deriving a bit $b$ from $B_{u'_j+1},B_{u'_j+2},\ldots, B_{u'_j+n}$.
\end{itemize}
\item \label{hybitem:endfair} {\em Decommitments phase.}
\begin{itemize}
\item For each $i\in[m]$: $P_i$ reclaims her $q$ coins by posting the decommitment $d_i$ in any of the blocks $B_{u''_j+1},B_{u''_j+2}\ldots,B_{u''_j+t}$.
\item Wait until the blockchain reaches block $B_{u''_j+t+k}$.
\item For each $i\in[m]$:
\begin{itemize}
\renewcommand{\labelitemiv}{$\star$}
\item If $P_i$ did not forfeit prior to round $j$, and posted her decommitment $d_i$ until the block $B_{u''_j+t}$, then $d'_i=d_i$.
\item Otherwise, $P_i$ forfeits and $d'_i=0$.
\end{itemize}
\item Set $s_j=b \oplus f(d'_1,d'_2,\ldots,d'_m)$, and set $u_{j+1}=u''_j+1$.
\end{itemize}
\end{itemize}
\item Output the bit $\mathsf{majority}(s_1,s_2,\ldots,s_r)$.
\end{itemize}
}
\end{framed}
\end{minipage}
\caption{The Hybrid beacon protocol (think $r=1,f=\mathsf{majority}$).}\label{fig:hyb}
\end{figure}
Let us examine the case of $r=1$ and $f(x_1,x_2,\ldots,x_m)=\mathsf{majority}(x_1,x_2,\ldots,x_m)$ first. Suppose for the moment that the adversary $\hybadv$ does not possess any mining power. Suppose that $m$ is odd, and denote $h\triangleq\frac{m+1}{2}$. In case $\hybadv$ corrupts a minimal majority of the parties, say $P_1,P_2,\ldots,P_h$, she can choose $d_1=d_2=\cdots=d_h=1$ in the commitments phase. % of $\hybprot$.
Then, $\hybadv$ will inspect the bit $b$ that was derived in beacon phase, and will either decommit all of the bits $d_1,d_2,\ldots,d_h$ so that the output will be $b\oplus 1$, or withhold $h'$ of the $h$ bits to cause the output to be $b$ at the cost of losing $h'\cdot q$ coins. We thus see that when $\hybadv$ corrupts a majority of the designated parties, she has complete control over the output. Note that $1\leq h' \leq h$, with $h'=1$ in the case that all of the $m-h$ honest parties committed to $0$, and $h'=h$ in the case that all of the $m-h$ honest parties committed to $1$. Since the probability that $\hybadv$ would wish to flip $b$ is $\nicefrac{1}{2}$, it follows that $\hybadv$ has to pay a penalty of more than $\frac{1}{2}q$ coins in expectation.\ariel{more than --$>$ at least, or alternatively explain why h' will be larger than 1 sometimes}

By contrast, consider a protocol that consists of only the first and third phases, i.e., without employing the beacon as the intermediate phase. An adversary who corrupts the majority of the designated parties has complete control over the output, without paying any penalty. This is because the adversary can simply commit to the output bit that she desires.

%If the penalty amount $q$ is set to a value that is significant relative to the perceived value of the output bit, it makes it less likely that $\hybadv$ will be able to persuade the designated parties to collude with her.
\smallskip

\ariel{`To elaborate' bad start for sentence, as it sounds like you are expanding on the scenario where there was no intermediate beacon phase}
The above observations imply that when 
%To elaborate, if
$\hybadv$ corrupts at least one (and at most $h$) of the designated parties, she can then either
\begin{itemize}
\item Expend resources during the beacon phase in an attempt to influence $b$ to be the opposite of her desired output, thus making it less likely that she will need to withhold some of her decommitments and lose $q$ or more coins.
\ariel{Might want to delete `to be the opposite of her desired output', this is assuming the comitted values give the opposite of her desired output} 
\item Or, remain idle during the beacon phase, and lose at least $q$ coins in case she wishes to flip $b$ during the decommitments phase.
\end{itemize}

By assuming that $\hybadv$ corrupts a minority of the designated parties, we can make the following claim.
\begin{claim}
Suppose that $\hybprot$ is parameterized according to $r=1$ and $f(x_1,x_2,\ldots,x_m)=\mathsf{majority}(x_1,x_2,\ldots,x_m)$. Suppose that the adversary $\hybadv$ purchased her mining power with $R_1$ \newpage coins, and has an additional reserve of $R_2$ coins. Assume that the following conditions hold:
\begin{enumerate}
\item $R_2+\maxcoins(R_1,n+2t+3k)<q$.
\item $\hybadv$ does not go under budget.
\item $\hybadv$ corrupts $m'<\frac{m}{2}$ of the designated parties.
\end{enumerate}
Then, the output of $\hybprot$ is less biased than the output of $\bprot$.
\end{claim}

\begin{proof}
The first condition implies that $\hybadv$ does not have enough funds to withhold any decommitment, unless she goes under budget. Therefore, the second condition implies that $\hybadv$ can influence the output only during the beacon phase. However, the third condition implies that $\hybadv$ will have an uncertainty regarding the direction in which to influence the beacon, because 
%she cannot be sure of the value of  the bit $f(d_1,\ldots,d_m)$ outputted by the designated parties that will be xored with the beacon output.
%This is the case since she controls less than half the parties
%and $f$ is the majority function.
the output may or may not become flipped when the other $m-m'$ designated parties reveal their input bits in the decommitments phase.
\ariel{Was initially confused by phrasing, changed it a bit (iddo: in comments now, will tweak soon)}
\end{proof}
%In particular, if $\hybadv$ fails to corrupt any of the designated parties, then the output of $\hybprot$ is a uniform random bit. 

Let us now consider $r=1$ and $f(x_1,x_2,\ldots,x_m)=x_1\oplus x_2\oplus\cdots\oplus x_m$. If at least one of the designated parties is honest, then it is useless for $\hybadv$ to try to influence the beacon phase, and instead she must lose the $q$ coins penalty if she wishes to influence the output bit. Thus, $\hybadv$ may be a more powerful adversary in the case of $f=\mathsf{majority}$, as then the mining power of $\hybadv$ can be effective even if she does not corrupt all of the designated parties. On the other hand, when $f=\mathsf{majority}$ and $\hybadv$ corrupts only one designated party, she may fail to influence the output if she withholds the decommitment. Hence, $f=\mathsf{majority}$ and $f=\mathsf{xor}$ are a priori incomparable. Under the reasonable assumption that the cost of influencing the beacon is high, it should be better to use $f=\mathsf{majority}$ in order to make it difficult for $\hybadv$ to influence $f$. In the case that the cost of influencing the beacon is low, it might be better to use $f=\mathsf{xor}$ in order to make it less likely that the distribution of the output bit is affected by the beacon phase (this can happen only if $\hybadv$ corrupts all the designated parties).
\ariel{wip: (Recall this can happen only if $\hybadv$ corrupts all the designated parties, as in the case $f=\mathsf{xor}$ one honest designated party suffices for a random output bit regargless of the beacon output).}
\ariel{I think there's a mix up in last sentence here. output is affected by beacon phase when all parites are corrupt}

\smallskip
Let us also consider the case of incorporating more rounds in $\hybprot$.
\begin{claim}
Suppose that $\hybprot$ is parameterized according to an odd $r > 1$ and $f(x_1,x_2,\ldots,x_m)=x_1\oplus x_2\oplus\cdots\oplus x_m$. Suppose that the adversary $\hybadv$ purchased her mining power with $R_1$ coins, and has an additional reserve of $R_2$ coins. Let $\eps>0$ be arbitrary. Denote $\ell\triangleq \floor{2\eps\frac{\pi}{\euler}\sqrt{r-\sqrt{r}}}-1$. Assume that the following conditions hold:
\begin{enumerate}
\item $R_2+\maxcoins(R_1,r(n+2t+3k))\leq\ell q$. 
\item $\hybadv$ does not go under budget.
\item $\hybadv$ corrupts $m'\leq m-1$ of the designated parties.
\end{enumerate}
Then, the protocol $\hybprot$ outputs a bit whose statistical distance from a uniform random bit is at most $\eps$.
\end{claim}
\begin{proof}
The third condition implies that $\hybadv$ must lose $q$ coins for every round that she wishes to control. Hence, the first and second conditions imply that $\hybadv$ would not have enough funds to control more than $\ell$ rounds by the time that the protocol terminated. We can thus regard $\hybadv$ as having $\ell'\leq\ell$ locations that she controls, and the remaining $r-\ell'$ locations are uniform random. While the majority extractor is an $\eps$-extractor for $(r,r-\ell)$ non-oblivious bit-fixing source, an inspection of Lemma \ref{lemma:majextr} shows that the same holds even in the case of a somewhat more powerful adversary. This adversary is a a quota of $\ell$ locations, and upon seeing the inputs $x_1,x_2,\ldots,x_{i-1}$ the adversary is allowed to decide whether to take control over $x_i$ by decrementing her quota, or let $x_i$ be sampled as a uniform random bit. The relevant observation is that the event $E$ in Lemma \ref{lemma:majextr} has the same meaning both in the case of an adversary that plays at $\ell$ locations that are fixed in advance, and an adversary who chooses adaptively the location in which she plays. It therefore follows that the statistical distance between the output bit and a uniform random bit is $\eps$ or less.
\end{proof}
\ariel{This assumes she has 0 influence in other rounds and their output is completely uniform. Don't understand why}
As discussed above, the only reason for employing the beacon when $\hybprot$ is parameterized according to $f(x_1,x_2,\ldots,x_m)=x_1\oplus x_2\oplus\cdots\oplus x_m$ is to handle the possibility that the adversary $\hybadv$ corrupts all the designated parties.

\smallskip
Let us note that under the assumption that $\hybadv$ is able to corrupt significantly less than half of the designated parties, the hybrid protocol $\hybprot$ can be improved by replacing $f$ with the {\em iterated majority} extractor \cite{RC:BL90,SICOP:KZ07}. Specifically, when the number of non-random inputs is $\approx m^{0.63}$ or less, the iterated majority extractor is strictly better than the majority extractor. For example, if $\hybadv$ corrupts one of $m=9$ designated parties, she will have $\nicefrac{{8\choose 4}}{2^8}=0.273$ probability to withhold and flip the beacon bit in the case of the majority extractor, and $\nicefrac{1}{2}\cdot\nicefrac{1}{2}=0.25$ probability in the case of the iterated majority extractor (cf. \cite[Figure 1]{FC:BGM16}).

\smallskip
In comparison to $\bprot$, a disadvantage of $\hybprot$ is that the public verifiability aspect is diminished. That is, even if all the designated parties decommitted, it cannot be verified whether they would have also decommitted (as opposed to taking a loss of at least $q$ coins) in case the output of the beacon phase was $b\oplus 1$ instead of $b$.

%Moreover, it is possible to incorporate additional rounds to further reduce the bias of the output, in a similarly to \cite[Section 4]{STOC:Cleve86}. A protocol $\hybprot'$ can do $r$ repeated invocations of $\hybprot$, such that the output $s_j$ of the $j^\text{th}$ invocation does not use the majority extractor, and is instead defined as $s_j=b \oplus d'_1 \oplus d'_2 \cdots \oplus d'_m$. In the case that a party $P_i$ withheld her decommitment in round~$j$, we define $d'_i=0$ for the remaining rounds $j,j+1,j+2,\ldots,r$. Finally, the output bit of $\hybprot'$ is defined as $\mathsf{majority}(s_1,s_2,\ldots,s_r)$, or as $\mathsf{iteratedmajority}(s_1,s_2,\ldots,s_r)$. If at least one of the designated parties is honest, then it is useless for $\hybadv$ to try to influence the beacon phases, and instead she must lose the $q$ coins penalty (and one of the corrupt parties) for each round that she wishes to control. If $\hybadv$ corrupts all of the $m$ parties, then in each round she can either try to influence the beacon or withhold a decommitment.

    %More sophisticated constructions as in \cite{TCC:MoranNS09,CRYPTO:BeimelOO10,STOC:HaitnerT14} may also be possible, depending on the limited expressibility that Bitcoin scripts allow.

%\newpage
\section{Practical considerations}\label{sec:practical}
For a beacon protocol that is dependent upon Bitcoin or similar systems, the real-world aspects of {\em proof of work} (PoW) based cryptocurrencies present some additional security concerns, as well as certain beneficial factors.
%have both negative and positive implications.

First, let us consider the additional potential risks. Our analysis assumed an adversary with $\beta$ mining power, such that the performance of the $\bprot$ protocol is better when $\beta$ is smaller. However, in the real world, Bitcoin miners delegate their power to centralized pools. We can thus regard the pool administrator as the adversary, as is it abnormal for a single pool to hold more than $\frac{1}{2}$ of the mining power. % (see, e.g., \cite{MAN:r1}).
Moreover, the risk that an adversarial pool administrator (with a large $\beta$) poses is less profound than what one may think. The reason for that is that the pool administrator can influence $\bprot$ only by discarding ``detrimental'' blocks (cf. Definition~\ref{def:helpdet}) that are solved by the miners who delegate their power to the pool. When the miners who participate in the pool would notice that their solved blocks are not being used and thus the rewards are not being added to the pool's reserve, they are likely to switch to %delegate their mining power to
a competing pool.

%Benefits: can take for example majority($\nicefrac{n}{2}$ Bitcoin blocks and $\nicefrac{n}{2}$ Litecoin blocks), to influence this the adversary needs to invest in different kinds of mining equipment.
On the other hand, the real-world behavior of the miners of crytocurrencies provides opportunities to strengthen the security of a beacon protocol. Since there are multiple cryptocurrency systems that enjoy some level of popularity, and these systems are based on {different} PoW hash functions, the beacon protocol can amplify its security by relying on multiple blockchains. The benefit in this stems from the observation that an adversary \adv who wishes to influence the beacon output would need invest in different kinds of mining equipment (or bribe the miners of different systems). That is, the amount of popularity that various cryptocurrency systems have can be harnessed to increase the overall security level of the beacon protocol. For example, let us make the following suppositions with regard to Bitcoin and Litecoin~\cite{MAN:LTC}:
\begin{enumerate}
\item Both Bitcoin and Litecoin have a similar market entry cost, implying that the profit margin for either of them is low.
\item The purchasing power of the Bitcoin currency is $c_1$ times greater than that of the Litecoin currency, where $c_1\geq 1$ because Bitcoin enjoys a greater level of adoption among the population.
\item The level of PoW based security that Bitcoin has is $c_2$ times greater than that of Litecoin, where $c_2\geq 1$ because Bitcoin is more popular and thus more miners have a vested interest to keep it secure.
\end{enumerate}
Notice that if \adv needs to invest $t$ coins to acquire $p<1$ %fraction
of the Bitcoin mining power, then the third supposition implies that \adv obtains $\approx c_2 \cdot p$ of the Litecoin mining power by investing $t$ coins into Litecoin mining equipment. Also note that $c_1\neq c_2$ is possible at least over the short term, though it is reasonable to assume that the adoption level would drive the security level over the long term. To estimate $c_2$, one should compare the cost (and availability) of mining equipment to the current PoW difficulty target of the Bitcoin and Litecoin networks.

Consider a protocol that uses $\mathsf{majority}(B_1,B_2,\ldots,B_{m},B'_1,B'_2,\ldots,B'_{w})$ to derives the output bit of the beacon, where $B_1,\ldots,B_{m}$ are Bitcoin blocks, $B'_1,\ldots,B'_{w}$ are Litecoin blocks, and the timestamps of $B_1$ and $B'_1$ are approximately the same. What would be a good choice for the parameter $w$?

Let $x$ denote the value of each Bitcoin block reward. Due to our second supposition, the value of each Litecoin block reward is $\nicefrac{x}{c_1}$. The adversary \adv is expected to solve $p\cdot m$ Bitcoin blocks, and thus lose $\approx x\cdot\frac{pm}{2}$ in value as she tries to influence the output bit. Also, \adv is expected to solve $c_2\cdot p\cdot w$ Litecoin blocks, and thus lose  $\approx \frac{x}{c_1}\cdot\frac{c_2 pw}{2}$ in value.
Given our first supposition, let us simplify further by assuming that the potential revenues that \adv can earn while creating blocks either among $B_1,\ldots,B_m$ or among $B'_1,\ldots,B'_w$ are relatively insignificant.
Therefore, in the case that \adv invested the same amount into Bitcoin mining equipment and Litecoin mining equipment, $w=\frac{c_1}{c_2} m$ implies that it is about as costly for \adv to influence the Bitcoin portion and the Litecoin portion of the inputs to the majority function.
%To take some concrete figures
E.g., if the price of Litecoin is $100$ times smaller than that of Bitcoin, and the security of the Litecoin network is $50$ times smaller than that of Bitcoin, then $w=2m$ is a rational choice.

The running time of this protocol is comparable to $\bprot$ with $n=m+\nicefrac{w}{4}$, because the blocks interval of Litecoin blocks is $4$ times shorter than that of Bitcoin. Hence, the advantage of this beacon protocol over $\bprot$ is that \adv invested $2t$ coins to try to influence the $m+w$ blocks, while she would have needed to invest only $t$ coins to try to influence $m+\nicefrac{w}{4}$ Bitcoin blocks. Furthermore, the need to acquire different kinds of mining equipment (\texttt{SHA256} ASIC for Bitcoin and \texttt{scrypt} ASIC for Litecoin) makes the task of the adversary more demanding.

%\newpage
%\input{bstat}

%\newpage
%\input{propertybased} 

\bigskip
\bigskip
\noindent{\bf Acknowledgments.} The first author thanks Aggelos Kiayias for many useful discussions. We also thank Hong-Sheng Zhou and Vassilis Zikas for useful discussions.

%\newpage

\ignore{
{%\footnote-size
\bibliographystyle{plain}
\bibliography{beacon,crypto/abbrev3,crypto/crypto,crypto/missing}

\begin{thebibliography}{10}

\bibitem{WWW:NISTBEACON}
{NIST} randomness beacon.
\newblock {\small\url{http://www.nist.gov/itl/csd/ct/nist_beacon.cfm}}.

\bibitem{C:AD15}
Marcin Andrychowicz and Stefan Dziembowski.
\newblock Distributed cryptography based on the proofs of work.
\newblock In {\em Crypto}, 2015.
\newblock {\small\url{http://eprint.iacr.org/2014/796}}.

\bibitem{SP:ADMM14}
Marcin Andrychowicz, Stefan Dziembowski, Daniel Malinowski, and Lukasz Mazurek.
\newblock Secure multiparty computations on bitcoin.
\newblock In {\em IEEE Security and Privacy}, 2014.

\bibitem{EPRINT:L13}
Lear Bahack.
\newblock Theoretical bitcoin attacks with less than half of the computational
  power (draft).
\newblock 2013.
\newblock {\small\url{http://eprint.iacr.org/2013/868}}.

\bibitem{MAN:BB02}
Harald Baier and Johannes Buchmann.
\newblock Generation methods of elliptic curves, 2002.
\newblock
  {\small{\url{https://www.ipa.go.jp/security/enc/CRYPTREC/fy15/doc/1030_Buchmann.evaluation.pdf}}}.

\bibitem{ICALP:BEG15}
Salman Beigi, Omid Etesami, and Amin Gohari.
\newblock Deterministic randomness extraction from generalized and distributed
  santha-vazirani sources.
\newblock In Magn{\'u}s~M. Halld{\'o}rsson, Kazuo Iwama, Naoki~Kobayashi 0001,
  and Bettina Speckmann, editors, {\em ICALP (1)}, volume 9134 of {\em Lecture
  Notes in Computer Science}, pages 143--154. Springer, 2015.

\bibitem{CRYPTO:BeimelOO10}
Amos Beimel, Eran Omri, and Ilan Orlov.
\newblock Protocols for multiparty coin toss with dishonest majority.
\newblock In Tal Rabin, editor, {\em CRYPTO}, volume 6223 of {\em Lecture Notes
  in Computer Science}, pages 538--557. Springer, 2010.

\bibitem{EPRINT:BFS16}
Mihir Bellare, Georg Fuchsbauer, and Alessandra Scafuro.
\newblock N{I}{Z}{K}s with an untrusted {C}{R}{S}: Security in the face of
  parameter subversion, 2016.
\newblock {\small\url{http://eprint.iacr.org/2016/372}}.

\bibitem{RC:BL90}
Michael Ben-Or and Nathan Linial.
\newblock Collective coin flipping.
\newblock In {\em Randomness and Computation}, pages 91--115, 1990.

\bibitem{SP:BCGTV15}
Eli Ben-Sasson, Alessandro Chiesa, Matthew Green, Eran Tromer, and Madars
  Virza.
\newblock Secure sampling of public parameters for succinct zero knowledge
  proofs.
\newblock In {\em IEEE Symposium on Security and Privacy}, pages 287--304. IEEE
  Computer Society, 2015.

\bibitem{FC:BGM16}
Iddo Bentov, Ariel Gabizon, and Alex Mizrahi.
\newblock Cryptocurrencies without proof of work.
\newblock In {\em Financial Cryptography Bitcoin Workshop}, 2016.
\newblock {\small\url{http://arxiv.org/abs/1406.5694}}.

\bibitem{C:BenKum14}
Iddo Bentov and Ranjit Kumaresan.
\newblock How to use bitcoin to design fair protocols.
\newblock In {\em Crypto (2)}, pages 421--439, 2014.

\bibitem{EPRINT:BLN15}
Daniel~J. Bernstein, Tanja Lange, and Ruben Niederhagen.
\newblock Dual {E}{C}: A standardized back door, 2015.
\newblock {\small\url{http://eprint.iacr.org/2015/767}}.

\bibitem{TCC:BCIPO13}
Nir Bitansky, Alessandro Chiesa, Yuval Ishai, Rafail Ostrovsky, and Omer
  Paneth.
\newblock Succinct non-interactive arguments via linear interactive proofs.
\newblock In {\em {TCC}}, pages 315--333, 2013.

\bibitem{ANTS:Boneh98}
Dan Boneh.
\newblock The decision diffie-hellman problem.
\newblock In {\em ANTS: 3rd International Algorithmic Number Theory Symposium
  (ANTS)}, 1998.
\newblock {\small\url{http://crypto.stanford.edu/~dabo/abstracts/DDH.html}}.

\bibitem{EPRINT:BCG15}
Joseph Bonneau, Jeremy Clark, and Steven Goldfeder.
\newblock On bitcoin as a public randomness source, 2015.
\newblock {\small\url{https://eprint.iacr.org/2015/1015}}.

\bibitem{WWW:browndg}
Daniel~G. Brown.
\newblock How i wasted too long finding a concentration inequality for sums of
  geometric variables.
\newblock {\small\url{https://cs.uwaterloo.ca/~browndg/negbin.pdf}}.

\bibitem{SP:Z14}
Alessandro Chiesa, Christina Garman, Ian Miers, Madars Virza, Eli Ben-Sasson,
  Matthew Green, and Eran Tromer.
\newblock Zerocash: Practical decentralized anonymous e-cash from bitcoin,
  2014.
\newblock IEEE Symposium on Security and Privacy (S\&P).

\bibitem{ClarkH10}
Jeremy Clark and Urs Hengartner.
\newblock On the use of financial data as a random beacon.
\newblock In {\em 2010 Electronic Voting Technology Workshop / Workshop on
  Trustworthy Elections, {EVT/WOTE} '10, Washington, D.C., USA, August 9-10},
  2010.

\bibitem{STOC:Cleve86}
Richard Cleve.
\newblock Limits on the security of coin flips when half the processors are
  faulty (extended abstract).
\newblock In {\em STOC}, pages 364--369, 1986.

\bibitem{unprofit1}
Gerald Davis.
\newblock Captain (not so) obvious: most asic miners will lose money (and why
  it is ok), 2013.
\newblock {\small\url{https://bitcointalk.org/index.php?topic=262473.0}}.

\bibitem{P2P:DW13}
Christian Decker and Roger Wattenhofer.
\newblock {Information Propagation in the Bitcoin Network}.
\newblock In {\em {13th IEEE International Conference on Peer-to-Peer Computing
  (P2P)}}, 2013.

\bibitem{DiffieHellman76}
W.~Diffie and M.~Hellman.
\newblock New directions in cryptography.
\newblock {\em IEEE Transactions on Information Theory}, 22(6):644--654,
  November 1976.

\bibitem{books:daglib0025902}
Devdatt~P. Dubhashi and Alessandro Panconesi.
\newblock {\em Concentration of Measure for the Analysis of Randomized
  Algorithms.}
\newblock Cambridge University Press, 2009.

\bibitem{FC:ES14}
Ittay Eyal and Emin~G{\"u}n Sirer.
\newblock Majority is not enough: Bitcoin mining is vulnerable.
\newblock In {\em Financial Cryptography}, 2014.

\bibitem{unprofit5}
Peter Farquhar.
\newblock How i nearly lost \$10,000 in my hamfisted struggle to mine bitcoin
  in australia, 2014.
\newblock
  {\small\url{http://www.businessinsider.com.au/how-i-nearly-lost-10000-in-my-hamfisted-struggle-to-mine-bitcoin-in-australia-2014-7}}.

\bibitem{EC:GKL15}
Juan~A. Garay, Aggelos Kiayias, and Nikos Leonardos.
\newblock The {B}itcoin backbone protocol: Analysis and applications.
\newblock In {\em Eurocrypt}, 2015.
\newblock {\small\url{http://eprint.iacr.org/2014/765}}.

\bibitem{C:GO07}
Jens Groth and Rafail Ostrovsky.
\newblock Cryptography in the multi-string model.
\newblock In {\em Advances in Cryptology - {CRYPTO} 2007, 27th Annual
  International Cryptology Conference, Santa Barbara, CA, USA, August 19-23,
  2007, Proceedings}, pages 323--341, 2007.

\bibitem{EC:GS08}
Jens Groth and Amit Sahai.
\newblock Efficient noninteractive proof systems for bilinear groups.
\newblock In {\em Eurocrypt}, 2008.
\newblock {\small\url{https://eprint.iacr.org/2007/155}}.

\bibitem{STOC:HaitnerT14}
Iftach Haitner and Eliad Tsfadia.
\newblock An almost-optimally fair three-party coin-flipping protocol.
\newblock In David~B. Shmoys, editor, {\em STOC}, pages 408--416. ACM, 2014.

\bibitem{CCC:IshaiKO07}
Yuval Ishai, Eyal Kushilevitz, and Rafail Ostrovsky.
\newblock Efficient arguments without short {PCP}s.
\newblock In {\em IEEE Conference on Computational Complexity}, pages 278--291.
  IEEE Computer Society, 2007.

\bibitem{C:IOZ14}
Yuval Ishai, Rafail Ostrovsky, and Vassilis Zikas.
\newblock Secure multi-party computation with identifiable abort.
\newblock In Juan~A. Garay and Rosario Gennaro, editors, {\em CRYPTO (2)},
  volume 8617 of {\em Lecture Notes in Computer Science}, pages 369--386.
  Springer, 2014.

\bibitem{SICOP:KZ07}
Jesse Kamp and David Zuckerman.
\newblock Deterministic extractors for bit-fixing sources and
  exposure-resilient cryptography.
\newblock In {\em SICOMP: SIAM Journal on Computing}, volume~36, 2007.

\bibitem{EC:KZZ15}
Aggelos Kiayias, Hong-Sheng Zhou, and Vassilis Zikas.
\newblock Fair and robust multi-party computation using a global transaction
  ledger.
\newblock In {\em Eurocrypt}, 2015.
\newblock {\small\url{http://eprint.iacr.org/2015/574}}.

\bibitem{SP:KMSWP15}
Ahmed~E. Kosba, Andrew Miller, Elaine Shi, Zikai Wen, and Charalampos
  Papamanthou.
\newblock Hawk: The blockchain model of cryptography and privacy-preserving
  smart contracts.
\newblock In {\em IEEE S\&P}, 2016.
\newblock {\small\url{http://eprint.iacr.org/2015/675}}.

\bibitem{CCS:KumBen14}
Ranjit Kumaresan and Iddo Bentov.
\newblock How to use bitcoin to incentivize correct computations.
\newblock In {\em CCS}, 2014.

\bibitem{MAN:LTC}
Charles Lee.
\newblock Litecoin.
\newblock {\small{\url{https://litecoin.org/}}}.

\bibitem{unprofit2}
Alec Liu.
\newblock A guide to bitcoin mining: Why someone bought a \$1,500 bitcoin miner
  on ebay for \$20,600, 2013.
\newblock
  {\small\url{http://motherboard.vice.com/blog/a-guide-to-bitcoin-mining-why-someone-bought-a-1500-bitcoin-miner-on-ebay-for-20600}}.

\bibitem{MAN:Mironov}
Ilya Mironov.
\newblock Hash functions: Theory, attacks, and applications.
\newblock 2005. {\em{Microsoft Research (MSR).}}
  {\small{\url{http://research.microsoft.com/pubs/64588/hash_survey.pdf}}}.

\bibitem{TCC:MoranNS09}
Tal Moran, Moni Naor, and Gil Segev.
\newblock An optimally fair coin toss.
\newblock In Omer Reingold, editor, {\em TCC}, volume 5444 of {\em Lecture
  Notes in Computer Science}, pages 1--18. Springer, 2009.

\bibitem{MAN:n08}
Satoshi Nakamoto.
\newblock Bitcoin: A peer-to-peer electronic cash system.
\newblock {\em Bitcoin.org}, 2008.

\bibitem{unprofit4}
Dario~Di Pardo.
\newblock \$46k spent on bitcoin mining hardware: The final reckoning, 2014.
\newblock
  {\small\url{http://www.coindesk.com/46k-spent-mining-hardware-final-reckoning/}}.

\bibitem{RSA:P14}
Rene Peralta.
\newblock The {NIST} randomness beacon, 2014.
\newblock
  {\small\url{http://www.rsaconference.com/writable/presentations/file_upload/asec-t07b-the-nist-randomness-beacon_final.pdf}}.

\bibitem{EPRINT:PW16}
C{\'e}cile Pierrot and Benjamin Wesolowski.
\newblock Malleability of the blockchain's entropy, 2016.
\newblock {\small\url{http://eprint.iacr.org/2016/370}}.

\bibitem{BLOG:RVW12}
Omer Reingold, Salil Vadhan, and Avi Wigderson.
\newblock No deterministic extraction from santha-vazirani sources - a simple
  proof, 2012.
\newblock
  {\small\url{https://windowsontheory.org/2012/02/21/no-deterministic-extraction-from-santha-vazirani-sources-a-simple-proof/}}.

\bibitem{MAN:r12}
Meni Rosenfeld.
\newblock Analysis of hashrate-based double-spending, 2012.
\newblock \url{http://arxiv.org/abs/1402.2009}.

\bibitem{FOCS:SanVaz84}
Miklos Santha and Umesh Vazirani.
\newblock Generating quasi-random sequences from slightly-random sources.
\newblock In {\em FOCS: IEEE Symposium on Foundations of Computer Science
  (FOCS)}, 1984.

\bibitem{FC:SSZ16}
Ayelet Sapirshtein, Yonatan Sompolinsky, and Aviv Zohar.
\newblock Optimal selfish mining strategies in bitcoin.
\newblock In {\em Financial Cryptography}, 2016.

\bibitem{MISC:CLTV}
Peter Todd.
\newblock Checklocktimeverify, 2014.
\newblock
  {\small\url{https://github.com/petertodd/bips/blob/checklocktimeverify/bip-checklocktimeverify.mediawiki}}.

\bibitem{MAN:W09}
Brent Waters.
\newblock {C}{S} 388{H} introduction to cryptography, 2009.
\newblock
  {\small{\url{http://www.cs.utexas.edu/~bwaters/classes/cs388h-fall12/LectureNotes.pdf}}}.

\bibitem{unprofit0}
Kristina Zucchi.
\newblock Is bitcoin mining still profitable?
\newblock
  {\small\url{http://www.investopedia.com/articles/forex/051115/bitcoin-mining-still-profitable.asp}}.

\end{thebibliography}
%\bibliography{../crypto/conf,../crypto/crypto,../crypto/extra}
}
}

%\vspace{-2mm}

%\begin{flushright}
%{\tiny revision 48}
%\end{flushright}

\end{document}

\appendix
\newpage
\input{construction}

\end{document}